\documentclass[11pt]{article}

\usepackage[margin=1in]{geometry}
\usepackage{amsmath,amsthm,amssymb,mathtools}
\usepackage{cite}
\usepackage{array}
\usepackage[font=small,labelfont=bf]{caption}

\usepackage[ruled,norelsize,vlined]{algorithm2e}
\DontPrintSemicolon

\usepackage{accents}

\usepackage{etoolbox} 

\renewcommand{\gets}{:=}

\newcommand{\Otilde}[1]{\Tilde{O}(#1)}
\newcommand{\Wtilde}[1]{\Tilde{\Omega}(#1)}

\newcommand{\alglabel}[1]  {\label{alg:#1}}
\newcommand{\algref}[1]    {Algorithm~\ref{alg:#1}}
\newcommand{\algreftwo}[2] {Algorithms~\ref{alg:#1} and~\ref{alg:#2}}
\newcommand{\seclabel}[1]  {\label{sec:#1}}
\newcommand{\secref}[1]    {Section~\ref{sec:#1}}
\newcommand{\secreftwo}[2] {Sections~\ref{sec:#1} and~\ref{sec:#2}}
\newcommand{\lemlabel}[1]  {\label{lem:#1}}
\newcommand{\lemref}[1]    {Lemma~\ref{lem:#1}}
\newcommand{\lemreftwo}[2] {Lemmas~\ref{lem:#1} and~\ref{lem:#2}}
\newcommand{\thmlabel}[1]  {\label{thm:#1}}
\newcommand{\thmref}[1]    {Theorem~\ref{thm:#1}}
\newcommand{\corlabel}[1]  {\label{cor:#1}}
\newcommand{\corref}[1]    {Corollary~\ref{cor:#1}}
\newcommand{\deflabel}[1]  {\label{def:#1}}

\newcommand{\defn}[1]      {\textbf{\emph{#1}}}
\newcommand{\id}[1]        {\ifmmode\mathit{#1}\else\textit{#1}\fi}

\newcommand{\poly}[1]      {\mathrm{poly}(#1)}
\newcommand{\card}[1]      {\left| #1\right|}
\newcommand{\ang}[1]       {\left<#1\right>}
\newcommand{\set}[1]       {\left\{#1\right\}}
\newcommand{\prob}[1]      {\Pr\set{#1}}
\newcommand{\ceil}[1]      {\lceil#1\rceil}
\newcommand{\floor}[1]     {\lfloor#1\rfloor}

\newcommand{\Succ}[2]      {R^+\ifblank{#1}{}{(#1,#2)}}
\newcommand{\Pred}[2]      {R^-\ifblank{#1}{}{(#1,#2)}}
\newcommand{\SuccD}[3]     {R^+_{#3}(#1,#2)}
\newcommand{\PredD}[3]     {R^-_{#3}(#1,#2)}
\newcommand{\concat}       {\mapsto}
\newcommand{\Bridge}[2]    {\id{Bridge}(#1,#2)}
\newcommand{\Anc}[2]       {\id{Anc}(#1,#2)}
\newcommand{\Desc}[2]      {\id{Desc}(#1,#2)}
\newcommand{\AncD}[3]      {\id{Anc}_{#3}(#1,#2)}
\newcommand{\DescD}[3]     {\id{Desc}_{#3}(#1,#2)}
\newcommand{\BridgeD}[3]   {\id{Bridge}_{#3}(#1,#2)}
\newcommand{\head}[1]      {\id{head}(#1)}
\newcommand{\tail}[1]      {\id{tail}(#1)}
\newcommand{\length}[1]    {\id{length}(#1)}

\newtheorem{theorem}{Theorem}[section]
\newtheorem{lemma}[theorem]{Lemma}
\newtheorem{corollary}[theorem]{Corollary}

\newtheorem{definition}[theorem]{Definition}

\newenvironment{closeitemize}
    {\begin{list}{$\bullet$}
    {
         \setlength{\itemsep}{-0.3\baselineskip}
     \setlength{\topsep}{0.15\baselineskip}
     \setlength{\parskip}{0pt}}}
    {\end{list}}

\newcounter{ccount}
\newenvironment{closeenum}
    {\begin{list}{\arabic{ccount}.}
    {\usecounter{ccount}\setlength{\itemsep}{-0.3\baselineskip}
     \setlength{\topsep}{0.15\baselineskip}
     \setlength{\parskip}{0pt}}}
    {\end{list}}

\begin{document}
\begin{titlepage}
  \title{Nearly Work-Efficient Parallel Algorithm for Digraph
    Reachability}
  \author{Jeremy T. Fineman\\
    Georgetown University\\
    \texttt{jfineman@cs.georgetown.edu}} \date{}
  \maketitle

  \begin{abstract}
    One of the simplest problems on directed graphs is that of
    identifying the set of vertices reachable from a designated source
    vertex.  This problem can be solved easily sequentially by
    performing a graph search, but efficient parallel algorithms have
    eluded researchers for decades.  For sparse high-diameter graphs
    in particular, there is no known work-efficient parallel algorithm
    with nontrivial parallelism.  This amounts to one of the most
    fundamental open questions in parallel graph algorithms: \emph{Is
      there a parallel algorithm for digraph reachability with nearly linear
      work?}  This paper shows that the answer is yes.

    This paper presents a randomized parallel algorithm for digraph
    reachability and related problems with expected work $\Otilde{m}$
    and span $\Otilde{n^{2/3}}$, and hence parallelism
    $\Wtilde{m/n^{2/3}} = \Wtilde{n^{1/3}}$, on any graph with $n$
    vertices and $m$ arcs.  This is the first parallel algorithm 
    having both nearly linear work and strongly sublinear span, i.e.,
    span $\Otilde{n^{1-\epsilon}}$ for any constant $\epsilon>0$. The
    algorithm can be extended to produce a directed spanning tree,
    determine whether the graph is acyclic, topologically sort the
    strongly connected components of the graph, or produce a directed
    ear decomposition, all with work $\Otilde{m}$ and span
    $\Otilde{n^{2/3}}$.
    
    The main technical contribution is an \emph{efficient} Monte Carlo
    algorithm that, through the addition of $\Otilde{n}$ shortcuts,
    reduces the diameter of the graph to $\Otilde{n^{2/3}}$ with high
    probability.  While both sequential and parallel algorithms are
    known with those combinatorial properties, even the sequential
    algorithms are not efficient, having sequential runtime
    $\Omega(mn^{\Omega(1)})$.  This paper presents a surprisingly simple
    sequential algorithm that achieves the stated diameter reduction and runs
    in $\Otilde{m}$ time.  Parallelizing that algorithm yields the
    main result, but doing so involves overcoming several other
    challenges. 
  \end{abstract}
  \thispagestyle{empty}
\end{titlepage}

\section{Introduction}\seclabel{intro}

There are essentially no good parallel algorithms known for the most basic
problems on general directed graphs, especially when the graph is
sparse. This paper yields several.

A good parallel algorithm should have polynomial parallelism and be
(nearly) work efficient.  The \defn{work} $W(n)$ of a
parallel algorithm on a size-$n$ problem is the total number of
primitive operations performed.  Ideally, the work of the parallel
algorithm should be similar to the best sequential running time
$T^*(n)$ known for the problem.  An algorithm is \defn{work efficient}
if $W(n) \in O(T^*(n))$ and \defn{nearly work efficient} if
$W(n) \in \Otilde{T^*(n)} = O(T^*(n)\cdot \poly{\log{n}})$, where
$\Tilde{O}$ hides logarithmic factors.\footnote{In addition to
  uncluttering the bounds, ignoring logarithmic factors is
  particularly convenient when comparing parallel algorithms --- the
  precise bounds depend on the specifics of the parallel model, but
  the bounds typically only vary by logarithmic factors
  (see~\cite{JaJa92} for discussion) --- allowing us to focus on the
  high-level discussion.}  (As a slight abuse of notation,
$\Otilde{1}$ is used to mean $O(\poly{\log{n}})$, where the $n$ should
be clear from context.)\footnote{The standard definition for soft-$O$
  is that $f(n) \in \Otilde{g(n)}$ if
  $f(n) \in O(g(n)\, \poly{\log g(n)})$.  This paper uses
  $f(n) \in \Otilde{g(n)}$ to mean $f(n) \in O(g(n)\, \poly{\log n})$,
  with the only relevant difference being the meaning of
  $\Otilde{1}$.}  The \defn{span} $S(n)$, also called \defn{depth}, of
a parallel algorithm is the length of the longest chain of sequential
dependencies.\footnote{Older PRAM literature often characterizes
  algorithms by a number of processors and parallel running time.
  Span here is generally equivalent to parallel time, and work
  corresponds to the product of processors and time.}  By Brent's
scheduling principle~\cite{Brent74}, such an algorithm can generally
be scheduled to run in $O(W(n)/p)$ time on $p \leq W(n)/S(n)$
processors; adding more processors beyond that point does not yield
asymptotic speedup.  The limit $W(n)/S(n)$ is called the
\defn{parallelism} of the algorithm; an algorithm is \defn{moderately
  parallel} if the parallelism is $\Omega(n^\epsilon)$, for some
constant $\epsilon > 0$, and \defn{highly parallel} if the span is
$\Otilde{1}$.  The goal is to achieve speedup with respect to the best
sequential algorithm, which is why work efficiency matters.  A nearly
work-efficient algorithm runs in $\Otilde{T^*(n)/p}$ time on
$p\leq W(n)/S(n)$ processors, but inefficient algorithms may require
enormous numbers of processors to beat the sequential algorithm.

\paragraph{Remark.} Aside from the context provided in this
introduction and high-level ideas, most of the paper does not require
any specific knowledge of parallel algorithms; the challenge lies in
producing an algorithm with properties amenable to parallelization.
Most implementation details are straightforward, so the parallel model
and implementation details are deferred to \secref{details}.

\paragraph{Problem and history.}
Perhaps the most basic problem on directed graphs is the single-source
reachability problem: given a directed graph $G=(V,E)$ and source
vertex $s\in V$, identify the set of vertices reachable by a directed
path originating from $s$.  Throughout, let $n=\card{V}$ be the number
of vertices and $m=\card{E}$ be the number of arcs, and for
conciseness of bounds assume that $m\in\Omega(n)$.  This problem has
simple sequential solutions: both breadth-first search (BFS) and
depth-first search (DFS) solve the problem in $O(m)$ time.  There are
two natural parallel algorithms for the reachability problem, which
seem to be folklore. See Table~\ref{table:comparison} for a
comparison.  Parallel transitive closure~\cite{JaJa92}, which amounts
to repeated squaring of the adjacency matrix, is highly parallel but
far from work efficient even for dense graphs.  Parallel BFS is
similar to sequential BFS, except that arcs from each layer (vertices
with the same distance) are explored in parallel.  Parallel BFS is
work efficient (see, e.g.,~\cite{KleinSu97}), but the span is
proportional to the diameter, which is $\Theta(n)$ in the worst case.
Both algorithms fall short of our goals, but they remain the state of
the art.

\begin{table}
  \begin{center}
    \small
    \setlength\extrarowheight{2pt}
    \begin{tabular}{|c|c|c|c|c@{\hspace{1em}}l|}\hline
      & Work & Span & \begin{minipage}{2.5cm}\begin{center}Nearly\\ work efficient?\end{center}\end{minipage} & \multicolumn{2}{c|}{\begin{minipage}{4cm} 
          \begin{center}Number of processors
          to achieve
          $\Otilde{n/k}$
          runtime, for $m\in\Theta(n)$\end{center}\end{minipage}}
      \\ \hline\hline
      parallel BFS & $O(m)$ & $\Otilde{n}$ & Yes &
                                                   \multicolumn{2}{c|}{Not
                                                   Possible unless $k=\Otilde{1}$} \\
      \hline parallel Trans. Closure & $\Otilde{M(n)}$ & $\Otilde{1}$ &
                                                                            No
                                                                                    &
                                                                                      $kM(n)/n
                                                                                      \gg
                                                                                      nk$
                   &
                                                                                                             for
                                                                                                             $k\leq
                                                                                                             n$\\ \hline
      Spencer's~\cite{Spencer97} & $\Otilde{m+n\rho^2}$ &
                                                        $\Otilde{n/\rho}$
                    & if $\rho=\Otilde{\sqrt{m/n}}$ & $k^3$ & for
                                                              $k\leq n$
      \\ \hline
      UY~\cite{UllmanYa91}$^*$ & $\Otilde{m\rho +
                                                  \rho^4/n}$ 
             &
               $\Otilde{n/\rho}$
                    & if $\rho=\Otilde{1}$ & $k^2$ & for $k \leq
                                                     n^{2/3}$ $^\dagger$ \\ \hline
      This paper$^*$ & $\Otilde{m}$ & $\Otilde{n^{2/3}}$ & Yes & $k$ &for
                                                                       $k\leq
                                                                       n^{1/3}$ 
      \\ \hline
    \end{tabular}\vspace{-1em}
  \end{center}
  \caption[XXX]{
    Comparison of parallel algorithms for single-source
    reachability.  Two of the algorithms are parameterized by
    $\rho$, $1\leq \rho \leq n$, which trades off work and span.
    $M(n)$ is the work of the best highly parallel $n \times n$ matrix multiplication, which is at
    least the current best sequential time of $O(n^{2.372869})$~\cite{LeGall14}. \\
    $^*$: the algorithm is randomized. Bounds are with high probability.\\
    $^\dagger$: for higher $k$, the dependence on $k$ becomes worse
    and more complicated to state.
  }\label{table:comparison}
\end{table}
 
The only other progress on general graphs are work/span tradeoffs.
Ullman and Yannakakis~\cite{UllmanYa91} raised the question over 25
years ago of whether it is possible to solve digraph reachability with
sublinear work without sacrificing work efficiency. Instead, their
algorithm~\cite{UllmanYa91}, henceforth termed UY, and Spencer's
algorithm~\cite{Spencer97} exhibit tradeoffs between work and span.
Though not originally described in the same terms, both algorithms can
be parameterized by a value $\rho$, $1\leq \rho \leq n$.
Table~\ref{table:comparison} summarizes the performance
bounds.\footnote{The work bound stated by Ullman and
  Yannakakis~\cite{UllmanYa91} is worse, for small $\rho$, than the
  bound displayed in Table~\ref{table:comparison}. The table shows the
  improved bound observed by Schudy~\cite{Schudy08}.}  For $\rho=1$,
both algorithms are a parallel BFS.  As $\rho$ increases, the span
decreases but the work increases.  When $\rho = n$, both algorithms
converge to transitive closure via regular $\Theta(n^3)$-work matrix
multiplication.  They differ for intermediate $\rho$.  Spencer's
algorithm is deterministic and, for sufficiently dense graphs, can be
nearly work efficient with moderate parallelism.  In contrast, UY is
randomized and never simultaneously work efficient and moderately
parallel, but it exhibits a better work/span tradeoff for sparse
graphs.

Other work focuses on either restricted graph classes or sequential
preprocessing.  Kao and Klein~\cite{KaoKl90} give an
algorithm for reachability on planar digraphs with $\Otilde{n}$ work
and $\Otilde{1}$ span.  Klein~\cite{Klein93} gives an algorithm that
preprocesses the graph in $O(np)$ sequential time, where $p\geq 1$ is
a parameter; after the preprocessing, reachability can be solved in
$O(m/p)$ time on $p$ processors.

\subsection{Shortcutting Approach and Contributions}

The high-level approach is intuitive: (1) reduce the diameter of
the graph through the addition of \defn{shortcuts}, or arcs whose
addition does not change the transitive closure of the graph; (2) run
parallel BFS on the shortcutted graph.  UY~\cite{UllmanYa91} fits this
general strategy (and parallel BFS and transitive closure are extreme
cases), but Spencer's algorithm~\cite{Spencer97} does not.  If the BFS
phase is to complete with $\Otilde{m}$ work, then the number of shortcuts
added must be limited to $\Otilde{m}$.

Ignoring the cost of computing the shortcuts, $O(n)$ shortcuts are
known to be sufficient to reduce the diameter of any graph to
$\Otilde{\sqrt{n}}$. UY~\cite{UllmanYa91} with $\rho=\sqrt{n}$, for
example, accomplishes this task.  Ignoring logarithmic factors, this
is the best diameter reduction known for general graphs using a linear
number of shortcuts.  As Hesse~\cite{Hesse03} shows, however, there
exist graphs requiring $\Omega(mn^{1/17})$ shortcuts to reduce their
diameter below $\Theta(n^{1/17})$.\footnote{Closing the gap between $n^{1/17}$
and $\sqrt{n}$ is an interesting open question but not addressed by this
paper.}    Hesse's lower bound implies a lower bound on reachability via
shortcutting: any nearly work-efficient algorithm must have
$\Wtilde{n^{1/17}}$ span.

The main technical challenge is to produce the shortcuts efficiently,
which is a challenge even ignoring parallelism.  There is no
$\Otilde{m}$-time \emph{sequential} algorithm known to reduce every
graph's diameter to $\Otilde{n^{1-\epsilon}}$, for any constant
$\epsilon > 0$.  For contrast, consider the most natural approach
(similar to UY~\cite{UllmanYa91}): sample $\sqrt{n}$ vertices, perform
a graph search from each, and add shortcuts between all related pairs
of samples.  It is straightforward to prove that this algorithm yields
a graph with
$O(\sqrt{n}\log n)$ diameter, but the running time of the $\sqrt{n}$
independent searches is $O(m\sqrt{n})$.

This paper has the following main contributions:
\begin{closeitemize}
\item (\secref{sequential}.) An $\Otilde{m}$-time sequential Monte
  Carlo algorithm that shortens the diameter of any graph to
  $\Otilde{n^{2/3}}$, with high probability, through the addition of
  $\Otilde{n}$ shortcuts.  
\item (\secreftwo{parallel}{details}.)  A Monte Carlo parallel algorithm having
  $\Otilde{m}$ work and $\Otilde{n^{2/3}}$ span that shortens the
  diameter of any graph to $\Otilde{n^{2/3}}$, with high probability,
  through the addition of $\Otilde{n}$ shortcuts.  
\item Applying the diameter reduction then parallel BFS yields a Las
  Vegas algorithm for single-source reachability with $\Otilde{m}$
  work and $\Otilde{n^{2/3}}$ span, with high probability.
\item (\secref{tree}.) An extension that finds a
  \emph{directed spanning tree} of $G$ rooted at source $s$, i.e., a
  tree rooted at $s$ including only arcs in $G$ and containing all
  vertices reachable from $s$.
\end{closeitemize}
Applying existing reductions yields the following Las Vegas randomized
parallel algorithms, both with $\Otilde{m}$ work and $\Otilde{n^{2/3}}$
span with high probability:
\begin{closeitemize}
\item An algorithm that identifies and sorts the strongly connected
  components of the graph.  (Use the new reachability algorithm in
  Schudy's algorithm~\cite{Schudy08}.)
\item 
  An algorithm that finds a directed ear decomposition of any strongly
  connected graph.  (Use the new directed spanning tree algorithm with
  Kao and Klein's algorithm~\cite{KaoKl90}.)
\end{closeitemize}

\begin{algorithm}[t]
\small
\caption{Sequential algorithm for shortcutting}\alglabel{seq1}
\SetKwFunction{foo}{SeqSC1}
\Indm\foo{$G=(V,E)$}\\
\Indp
  \nl\lIf{$V = \emptyset$}{\KwRet{$\emptyset$}}
  \nl select a pivot $x\in V$ uniformly at random\\
  \nl let $\Succ{}{}$ denote the set of vertices reachable from $x$\\
  \nl let $\Pred{}{}$ denote the set of vertices that can reach $x$ \\
  \nl $S \gets \set{(x,v)| v\in \Succ{}{}} \cup
      \set{(u,x)|u\in\Pred{}{}}$ \tcp*{add shortcuts to/from vertices\\
    having paths from/to $x$, resp.\ }
  \nl $V_B \gets \Succ{}{} \cap \Pred{}{}$ ; \hspace{1em}
   $V_S \gets \Succ{}{} \backslash V_B$ ; \hspace{1em}
   $V_P \gets \Pred{}{} \backslash V_B$ ; \hspace{1em}
   $V_R \gets V \backslash (V_B \cup V_S \cup V_P)$\\ 
  \nl \Return $S \cup \text{\foo{$G[V_S]$}} \cup \text{\foo{$G[V_P]$}}
  \cup \text{\foo{$G[V_R]$}}$
\end{algorithm}

\subsection{Algorithm and Analysis Overview}
  
The sequential algorithm is simple enough that the main
subroutine is given immediately.  (See also \algref{seq1}.)  The
algorithm is recursive.  First select a random vertex $x$, called the
\defn{pivot}.  Perform a graph search forwards and backwards from $x$
to identify subsets $\Succ{}{}$ and $\Pred{}{}$, respectively.  Add
shortcuts from $\Pred{}{}$ to $x$ and from $x$ to $\Succ{}{}$.  The
graph is next partitioned into four subsets of vertices: $V_B$
comprises vertices in both $\Succ{}{}$ and $\Pred{}{}$, $V_S$
comprises vertices in $\Succ{}{}$ but not $\Pred{}{}$, $V_P$ comprises
vertices in $\Pred{}{}$ but not $\Succ{}{}$, and $V_{R}$ is all
remaining vertices.  Recurse on the subgraphs induced by the three
subsets $V_P$, $V_S$, and $V_R$.

Ignoring the addition of shortcuts, \algref{seq1} is essentially the
divide-and-conquer algorithm for topologically sorting the strongly
connected components of a graph described by Coppersmith et
al.~\cite{CoppersmithFlHe05}. Their proof thus carries over to prove
that this algorithm runs in $O(m\log n)$ sequential time in
expectation, but they do not address the diameter problem.

What should be surprising is that \algref{seq1} reduces the graph's
diameter, captured by the following lemma.  The proof is not obvious
and leverages new
insights and techniques.
\begin{lemma}\lemlabel{seqmain}
  Let $G=(V,E)$ be a directed graph, and consider any vertices $u,v\in
  V$ such that there exists a directed path from $u$ to $v$ in $G$.  Let $S$
  be the shortcuts produced by an execution of \algref{seq1}.  Then
  with probability at least $1/2$ (over random choices in \algref{seq1}), there exists a directed path from $u$ to $v$
  in $G_S = (V,E\cup S)$ consisting of $O(n^{2/3}\log^{4/3} n)$ arcs. 
\end{lemma}
\noindent As a corollary, through a simple application of a
Chernoff bound and union bound across $\leq n^2$ related pairs, the
union of shortcuts across $O(\log n)$ independent executions of
\algref{seq1} reduces the diameter of the graph to $\Otilde{n^{2/3}}$
with high probability.

\paragraph{Unusual aspects and insight.}  The analysis focuses on shortcutting a
particular path.  But unlike most divide-and-conquer analyses, the
division step here does not seem to effect progress.  Partitioning a
graph is good for reducing the problem size (which is what Coppersmith
et al.~\cite{CoppersmithFlHe05} leverage), but it is not good for
preserving paths --- and once vertices fall in different subproblems,
there can be no subsequent shortcuts between them.  This feature is
likely why previous algorithms, such as UY~\cite{UllmanYa91}, perform
independent searches on the original graph.

A key insight in the analysis is that the partitioning step also
reduces by a constant factor the number of vertices that could cause
the path to split again later.  In doing so, the probability of splitting
the path goes down, and hence the probability of shortcutting it goes
up.  The end effect is that the path is likely to be significantly
shortcutted before it is divided into too many pieces.

The proof of this filtering insight (\lemref{asymmetry}) leverages
antisymmetric relationships between certain vertices.  Interestingly,
the lack of symmetry in directed graphs is exactly the feature that
makes good parallel algorithms for digraphs so elusive, but here
asymmetry is crucial to the proof.

\paragraph{Building a parallel algorithm.} The main obstacle to
parallelizing \algref{seq1} is the graph searches employed to find
$\Succ{}{}$ and $\Pred{}{}$.  In fact, these searches are exactly the
single-source reachability problem that we want to solve.  The obvious
solution to try is to instead limit the searches to a distance of
$\Otilde{n^{2/3}}$, but unfortunately doing so causes other problems.
The parallel algorithm and the analysis are thus more involved.
\secref{parallel} provides a sequential algorithm with
distance-limited searches.  Given that, the parallel implementation
(\secref{details}) is straightforward.

\section{Preliminaries}\seclabel{prelim}

This section provides definitions, notations, and the main
probabilistic tools used throughout.

The subgraph of $G=(V,E)$ induced by vertices $V'\subseteq V$ is
denoted by $G[V']$.

If there is a directed path (possibly empty) from $u$ to $v$ in
digraph $G=(V,E)$, then \defn{$u$ precedes $v$} and \defn{$v$ succeeds
  $u$}, denoted $u\preceq v$.  We say also that $u$ \defn{can reach}
$v$ and that \defn{$v$ can be reached by $u$}.  If $u \preceq v$
and/or $v \preceq u$, then $u$ and $v$ are \defn{related}; otherwise
they are \defn{unrelated}.  The \defn{successors} or \defn{forward
  reach of $x$} is the set of nodes
$\Succ{G}{x} = \set{v | x \preceq v}$.  The \defn{predecessors} or
\defn{backwards reach of $x$} is the set
$\Pred{G}{x} = \set{u | u\preceq x}$.



A \defn{shortcut} is any arc $(u,v)$ such that $u\preceq v$ in $G$.

\paragraph{Paths and nonstandard notation.}  The analysis considers
paths as well as the relationships between paths and vertices.  A path
$P = \ang{v_0,v_1,\ldots,v_\ell}$ is denoted by the sequence of its
constituent vertices, with the arcs between consecutive pairs implied.
The first and last vertex of the path are denoted by $\head{P}$ and
$\tail{P}$, and the \defn{length} of the path, denoted $\length{P}$,
is the number of arcs.  For the path $P$ given, $\head{P} = v_0$,
$\tail{P}=v_\ell$, and $\length{P} = \ell$.  Two (possibly empty)
disjoint paths $P_1$ and $P_2$ may be concatenated, denoted
$P_1\concat P_2$, as long as the arc $(\tail{P_1},\head{P_2})$ exists.
\defn{Splitting a path $P$ into $k$ pieces} means partitioning it into
subpaths $P_1,P_2,\ldots,P_k$ such that
$P = P_1 \concat P_2 \concat \cdots \concat P_k$.

A vertex $x$ and a path $P$ can be compared in the following ways.
The vertex $x$ is a \defn{bridge of $P$} if $x$ can reach and can be
reached by vertices on the path, i.e., if there exists $v_i,v_j \in P$
such that $v_i \preceq x$ and $x \preceq v_j$. Note that every vertex
on the path is a bridge.  A vertex $x$ is an \defn{ancestor of $P$} if
$x$ can reach some vertex on the path, but $x$ cannot be reached by
any vertex on the path.  Similarly, $x$ is a \defn{descendent of $P$}
if $x$ can be reached by some vertex on the path, but $x$ cannot reach
any vertex on the path.  The set of all bridges, ancestors, and
descendents of $P$ are denoted $\Bridge{G}{P}$, $\Anc{G}{P}$, and
$\Desc{G}{P}$, respectively.  Note that these sets are all disjoint by
definition.  If a vertex $x$ is a bridge, ancestor, or descendent of
the path $P$, then $x$ and $P$ are \defn{related}.  Otherwise, they
are \defn{unrelated}.


\paragraph{Tools.}
The analysis employs one relatively uncommon probabilistic tool --- a special
case of Karp's~\cite{Karp94} probabilistic recurrence relations,
restated next. Roughly speaking, this theorem relates two processes:
(1) a random process where in each round the problem ``size'' ($\Phi$
in the theorem) reduces by a constant factor in expectation, and (2) a
deterministic process where the problem size reduces by exactly that
constant factor.  The theorem says that if the random process uses a
few extra rounds, it is very likely to experience at least the size
reduction of the deterministic process.
\begin{theorem}[Restatement of special case of Theorem~1.3\footnote{Karp states the theorem very differently. The
  process described here corresponds to his recurrence $T(I)
  = a(\Phi(I)) + T(h(I))$, where $a(x) = 0$, $x < d$ and $a(x)=1$,
  $x\geq d$, for $d = p^k \cdot \Phi(I_0)$. This recurrence counts the
  number of steps to reach the target size. (Note that $d$ depends only on
  the initial instance and is constant in the recurrence.) The deterministic
  counterpart is $\tau(x) = a(x) + \tau(px)$, which has solution $u(\Phi(I_0)) =
  \ceil{\log_{1/p}(\Phi(I_0)/d)} \leq  k+1$.} in
  \cite{Karp94}]\thmlabel{karp}
  Consider a random process of the following form.  Let $\cal I$
  denote the set of all problem instances, and let $I_0\in \cal I$
  denote the initial problem instance.  In the $r$th round, the
  process makes random choices and transforms the instance from 
  $I_{r-1}$ to $I_r$ (a random variable).  Let
  $\Phi : {\cal I} \rightarrow \mathbb{R}$ be any function satisfying
  $0 \leq \Phi(I_{r}) \leq \Phi(I_{r-1})$ for all relevant $r\geq 1$ and all
  feasible sequences $I_0,I_1,I_2,\ldots$ of instance outcomes.

  Suppose there exists some constant $p<1$ such that
  $E[\Phi(I_{r}) | I_0,I_1,\ldots,I_{r-1}] \leq p \cdot
  \Phi(I_{r-1})$, and consider any integers $k\geq 0$ and $w\geq 0$.
  Then $\prob{\Phi(I_{k+w+2}) > p^k \cdot \Phi(I_0)} \leq p^w$.
\end{theorem}

\section{Sequential Diameter Reduction}\seclabel{sequential}

This section focuses on proving the following theorem.  The unmodified
$G$ is used to refer to subgraphs $G=(V,E)$.  When the original input
graph is intended, $\hat{G}$ is employed instead.  Throughout, $x$
denotes the pivot, and the vertex sets $V_B$, $V_S$, $V_P$, and $V_R$
are used with meaning as setup in \algref{seq1}.
\begin{theorem}\thmlabel{seqfull}
  There exists a randomized sequential algorithm that takes as input a directed
  graph $\hat{G}=(\hat{V},\hat{E})$ and has the following guarantees, where
  $n=\card{\hat{V}}$, $m=\card{\hat{E}}$, and without loss of generality
  $m\geq n/2$: (1) the running time is $O(m\log^2n)$, (2) the
  algorithm produces a size-$O(n\log^2n)$ set $S^*$ of shortcuts, and
  (3) with high probability\footnote{With high probability means the
    failure probability can be driven down to $1/n^c$ for any constant
    $c$ by increasing the constants hidden inside the big-$O$ notation
    (specifically the running time and number of shortcuts here).}, the
  diameter of $G_{S^*}=(V,E\cup S^*)$ is $O(n^{2/3}\log^{4/3} n)$.
\end{theorem}
\noindent
As mentioned in \secref{intro}, the algorithm entails taking the union
of shortcuts from $\Theta(\log n)$ runs of \algref{seq1}.  To make the
running time worst case, there will be one minor modification
introduced later: namely, an
extra base case to truncate the recursion. 

\secreftwo{seqover}{seqsimple} set up the main ideas for proof of
\lemref{seqmain} but instead proves a weaker distance bound of
$O(n^{1/\lg(8/3)}) = O(n^{0.7067})$.  \secref{seqbetter} tightens the
distance bound to $O(n^{2/3}\log^{4/3} n)$, thereby proving \lemref{seqmain}.  It
is worth emphasizing that \secreftwo{seqsimple}{seqbetter} use exactly
the same algorithm --- the only difference is the details of the
analysis.  Finally, \secref{seqfull} completes the proof of
\thmref{seqfull} by analyzing the running time and number of
shortcuts.

\subsection{Setup of the Analysis}\seclabel{seqover}

Fix any simple path $\hat{P}=\ang{v_0,\ldots,v_\ell}$ in the graph up
front.  By partitioning the graph, each call to \foo also splits the
path into subpaths.  The analysis tracks a collection of calls whose
subgraphs contain subpaths of $\hat{P}$.  

More precisely, a \defn{path-relevant subproblem}, denoted by pair
$(G,P)$, corresponds to a call \foo{$G$} and an associated nonempty
subpath $P$ of $\hat{P}$ to shortcut.  The starting subproblem is
$(\hat{G},\hat{P})$.  The path-relevant subproblems are most
subproblems for which $G \cap \hat{P} \neq \emptyset$, except with a
base case occurring when a subpath $P$ is shortcutted to two hops ---
all recursive subproblems arising beyond that point are \emph{not}
path relevant.  The following lemma characterizes the path-relevant
subproblems that arise when executing the call \foo{$G$} with
associated path $P$.

It is worth emphasizing that the algorithm has no knowledge of the
path $P$; associating the subpath with the subproblem is an analysis
tool only.  


\begin{lemma}\lemlabel{rewrite}
  Let $P=\ang{v_0,\ldots,v_\ell}$ be a nonempty path in $G = (V,E)$,
  and consider the effect of a single call \foo{$G$} in \algref{seq1}.
  The following are the outcomes depending on pivot~$x$:
  \begin{closeenum}
    \item (Base case.)  If $x$ is a bridge of $P$, then the shortcuts
      $(v_0,x)$ and $(x,v_\ell)$ are created.  There are no
      path-relevant subproblems.  
    \item If $x$ and $P$ are unrelated, then $P$ is entirely contained
      in $G[V_R]$; the one path-relevant subproblem is thus $(G[V_R],P)$.
    \item If $x$ is an ancestor of $P$, then $P = P_1 \concat P_2$ for
      $P_1 = P \cap G[V_R]$ and $P_2 = P \cap G[V_s]$.  There are
      thus at most two path relevant
      subproblems: if $P_1$ is nonempty, $(G[V_R],P_1)$ is path
      relevant; if $P_2$ is nonempty, $(G[V_S],P_2)$
      is path relevant. 
    \item If $x$ is a descendent of $P$, then $P = P_1 \concat P_2$ for
      $P_1 = P \cap G[V_P]$ and $P_2 = P \cap G[V_R]$.  This case
      gives rise to at most two path-relevant subproblems, as above.
    \end{closeenum}
\end{lemma}
\begin{proof} The proof follows from the definitions. Consider for
  example the last case, that $x$ is a descendent of $P$. Then there
  is some latest vertex $v_k$ on the path such that $v_k \preceq x$.
  Then consider subpaths $\ang{v_0,\ldots,v_k}$ and
  $\ang{v_{k+1},\ldots,v_\ell}$. For all $v_i$ with $i\leq k$, we have
  $v_i \preceq v_k \preceq x$, and hence $P_1 = \ang{v_0,\ldots,v_k}$
  is entirely in $V_P$.  All $v_j$ with $j > k$ are unrelated to $x$
  and hence in~$V_R$.  
\end{proof}


Cases~3 and~4 seem like bad cases because the number of path-relevant
subproblems, and hence arcs in the final path, increases.
\secref{seqsimple} argues that these cases do make progress.

The path-relevant subproblems that arise during the execution of the
algorithm induce a \defn{path-relevant subproblem tree}, where each
node $s$ corresponds to a call of \foo on some path-relevant
subproblem $s = (G,P)$.  For the analysis, it is convenient to
consider the \defn{flattened path-relevant tree}, where each node
corresponding to case~2 in \lemref{rewrite} is merged with its only
child.  Viewed algorithmically, a node in the flattened path-relevant
tree corresponds to sampling multiple pivots $x$ (and discarding some
of the graph) until finally getting one that is related to the path
$P$.

The analysis considers levels in the flattened path-relevant tree in
aggregate, i.e., executing the algorithm in a breadth-first fashion.
The point is to later fit the analysis to \thmref{karp}.  Specifically, the
analysis consists of a sequence of rounds, where the instance $I_r$ in
round $r$ is the collection of subproblems defined by the nodes at
depth $r$ in the flattened path-relevant tree.  We have the following lemma
immediately.  All that remains is bounding the lengths (\secref{seqsimple}).

\begin{lemma}\lemlabel{tree}
  Consider any graph $\hat{G}=(\hat{V},\hat{E})$ and any path
  $\hat{P}$ from $u$ to $v$.  Consider an execution of \algref{seq1},
  let $S$ be the shortcuts produced, and let
  $\set{(G_1,P_1),\ldots,(G_k,P_k)}$ denote the set of path-relevant
  subproblems at level/depth $r$ in the flattened path-relevant
  tree. Then there is a $u$-to-$v$ path in
  $G_S = (\hat{V},\hat{E}\cup S)$ of length at most
  $2^r + 2^{r-1}+ \sum_{i=1}^k\length{P_i}$.
\end{lemma}
\begin{proof}
  Let $L_i$ denote the set of paths associated with leaves in the tree
  at depth $i$.  Then a simple induction over levels proves that: the
  set of paths
  $\set{P_1,\ldots,P_k} \cup \left(\bigcup_{i=1}^{r-1}L_i\right)$
  constitute a splitting (partition) of path $\hat{P}$.  To perform the
  inductive step, apply \lemref{rewrite} at each internal node.
  
  It remains to bound the path-length in $G_S$ by positing a specific
  path: the concatenation of the shortcutted paths for the leaves and
  the full unshortcutted paths for the remaining subproblems.  Each
  concatenation adds $1$ arc, each leaf's path uses $2$ shortcuts, and
  each remaining non-leaf path $P_i$ has $\length{P_i}$ arcs.  Since
  the degree of each node is at most 2 (\lemref{rewrite}), the number
  of leaves above level $r$ is at most $2^{r-1}$, and the number of
  internal nodes (concatenations) above level $r$ is also is at most
  $2^{r-1}$.  Adding everything together gives the bound.
\end{proof}
  
\subsection{Asymmetry Leads to Progress}\seclabel{seqsimple}

This section proves that with probability at least $1/2$, the distance
between $u$ and $v$ is at most $O(n^{1/\log(8/3)})$.  The main tools
are \thmref{karp} and a proof that the number of path-related vertices
decreases by a constant fraction, on average, with each level in the
flattened path-relevant tree.  More precisely, a vertex $v$ is \defn{path
  active at level~$r$} if (1) $v$ is part of some path-relevant
subproblem at level $r$ in the flattened tree, and (2) $v$ is related to the
path in that subproblem.  The goal is to argue that the expected
number of path-active vertices decreases with each level.  

Recall that the each node in the flattened tree corresponds to a
sequence of calls to \foo, where the last call happens to draw a pivot
$x$ that is path related.  The analysis focuses on that last choice
of~$x$. But consider instead the equivalent process of choosing $x$ by
first tossing a weighted coin to determine whether $x$ is bridge,
ancestor, or descendent, then choosing the specific vertex from within
that set uniformly at random.  The following lemma considers the
effect of choosing $x$ from all path ancestors.  Choosing from path
descendents is symmetric.

\begin{lemma}\lemlabel{asymmetry}
  Consider any subproblem $(G,P)$. Suppose that $x$ is drawn uniformly
  at random from $\Anc{G}{P}$, let $\alpha=\card{\Anc{G}{P}}$, and let $\alpha'$
  be denote the number of vertices in $\Anc{G}{P}$ that are path
  active at the next level, i.e., after calling \foo{$G$}.  Then
  $E[\alpha' | x \in \Anc{G}{P}] < \alpha/2$.
\end{lemma}
\begin{proof}
  Define the following binary relation over vertices in $\Anc{G}{P}$:
  \defn{$u$ preserves $v$} means that if $x=u$, then $v$ remains path
  active.  The relation is irreflexive by virtue of the fact that
  $x \in (\Pred{G}{x} \cap \Succ{G}{x})$ and hence not contained in
  any subproblems.  The goal is to prove that it is also
  antisymmetric.  Assuming the asymmetry, the total number of pairs
  satisfying the preserves relation is at most ${\alpha \choose 2}$.
  The number of vertices preserved by $x$ is $\alpha'$, and hence
  $E[\alpha'] \leq {\alpha \choose 2} / \alpha = (\alpha - 1)/2$.
  It remains to prove that the preserves relation is antisymmetric.  

  Consider any pair with $u$ preserves $v$, let
  $P=\ang{v_0,v_1,\ldots,v_\ell}$, and let $v_k$ be the earliest
  vertex in $P$ with $u \preceq v_k$.  Then selecting $x=u$ splits the
  path into $P_1^{(u)} = \ang{v_0,\ldots,v_{k-1}}$ and
  $P_2^{(u)} = \ang{v_k,\ldots,v_\ell}$, as stated in
  \lemref{rewrite}, where the superscript $^{(u)}$ indicates
  $x=u$. There are two ways that $v$ could be preserved: either
  $v \in V_R^{(u)}$ and $v \in \Anc{G[V_R^{(u)}]}{P_1^{(u)}}$, or
  $v \in V_S^{(u)}$ and $v \in \Anc{G[V_S^{(u)}]}{P_2^{(u)}}$.  (No relationships
  are added in the subproblems, so $v$ cannot, e.g., become a path
  descendent.)  Suppose the latter is true.  Then, by definition of
  $V_S^{(u)}$ in \algref{seq1}, $u \preceq v$ and $v \not\preceq u$.  It
  follows that if $x=v$, $u \in V_P^{(v)}$ and hence $v$ does not preserve
  $u$.

  Suppose instead that $v \in V_R^{(u)}$, implying $u$ and $v$ are
  unrelated.  Then $v$ can only preserve $u$ if
  $u\in \Anc{G[V_R^{(v)}]}{P_1^{(v)}}$.  But
  $v \in \Anc{G[V_R^{(u)}]}{P_1^{(u)}}$ implies $P_1^{(v)}$ is a
  subpath of $P_1^{(u)}$, which is unrelated to $u$.  So
  $u \not\in \Anc{G}{P_1^{(v)}} \supseteq
  \Anc{G[V_R^{(v)}]}{P_1^{(v)}}$, and hence $v$ does not preserve $v$.
\end{proof}

\lemref{asymmetry} states that if an ancestor is selected, the number
of path-active ancestors decreases by half.  The following lemma extends the
analysis to consider the effect on the total number of path-active
vertices.  The worst case is that the number of ancestors equals the
number of descendents.  Then \lemref{asymmetry} indicates that half of
the vertices decrease by half, i.e., a $3/4$ reduction in total.  

\begin{lemma}\lemlabel{nodereduction}
  Let $\eta$ denote the number of path-active vertices in some
  level-$(r-1)$ subproblem $(G,P)$, and let $\eta'$ be a random variable
  denoting the number of those vertices that are path active at
  level-$r$.  Then $E[\eta'] < (3/4)\eta$.
\end{lemma}
\begin{proof}
  Let $\alpha$, $\beta$, and $\delta$ denote the number of ancestors,
  bridges, and descendents, respectively, of path $P$ in $G$, with
  $\alpha+\beta+\delta = \eta$.  Selecting a pivot $x$ that is
  unrelated to $P$ can only decreases the number of path-related
  vertices, decreasing $\eta$ even further, so these choices can be
  ignored.  Consider the first $x$ that is related to $P$. If, for
  example, that $x$ is a path ancestor, \lemref{asymmetry} states
  $E[\eta' | x\in\Anc{G}{P}] < \alpha/2 + \beta + \delta$.  If $x$ is
  a bridge, then $\eta' = 0$ because there are no path-relevant
  subproblems.  Adding up all three cases and scaling by their
  probabilities, we have
  \begin{align*}
    E[\eta'] &= \left(\frac{\alpha}{\eta}\right) E[\eta' | x \in \Anc{G}{P}] +
    \left(\frac{\delta}{\eta}\right) \cdot E[\eta' | x \in
               \Desc{G}{P}]  \\ 
    &< \left(\frac{\alpha}{\eta}\right)(\alpha/2+\beta+\delta)  +
  \left(\frac{\delta}{\eta}\right)(\alpha+ \beta+\delta/2) \tag*{(\lemref{asymmetry})}\\
    &= \frac{(\alpha+\delta)(\alpha/2+\beta+\delta/2)}{\eta} +
      \frac{\alpha\delta}{\eta} \\
    &= \frac{(\eta-\beta)(\eta+\beta)}{2\eta} +
      \frac{(\sqrt{\alpha\delta})^2}{\eta} \tag*{($\eta=\alpha+\beta+\delta$)}\\
    &\leq \frac{\eta^2}{2\eta} + \frac{((\alpha+\delta)/2)^2}{\eta} \\
    &\leq (3/4)\eta \ . \\[-3em]
  \end{align*}
\end{proof}

For subproblem $s = (G,P)$, define $\phi(s)$ to be the number of
path-active vertices in $s$.  Define
$\Phi(I_r) = \sum_{s \in I_r} \phi(s)$, where $I_r$ is the collection
of subproblems at level-$r$ in the flattened tree.  Then we have the
following.  Applying \thmref{karp} then gives the main lemma.

\begin{corollary}\corlabel{levelreduction}
  Given any collection $I_{r-1}$ of subproblems,
  $E[\Phi(I_r) | I_{r-1}] \leq (3/4)\Phi(I_{r-1})$.
\end{corollary}
\begin{proof}
  \lemref{nodereduction} states that for each $s \in I_r$, we have
  $E[\phi(s_1) + \phi(s_2)] \leq (3/4)\phi(s)$, where $s_1$ and $s_2$
  are random variables for the at most two path-relevant subproblems
  of $s$.  The claim follows by linearity of expectation over all~$s$.
\end{proof}

\begin{lemma}\lemlabel{weakdiam}
  Let $\hat{G}=(\hat{V},\hat{E})$ be a directed graph, and consider any
  vertices $u,v\in V$ such that there exists a directed path from $u$
  to $v$ in $\hat{G}$.  Let $S$ be the shortcuts produced by an execution of
  \algref{seq1} and let $n=\card{\hat{V}}$.  Then with probability at least
  $1/2$, there exists a directed path from $u$ to $v$ in
  $G_S = (\hat{V},\hat{E}\cup S)$ consisting of $O(n^{1/\lg(8/3)})$ arcs.
\end{lemma}
\begin{proof}
  Choose an arbitrary simple path $\hat{P}$ from $u$ to $v$ in $\hat{G}$.
  At most every vertex is path active, so $\Phi(I_0) \leq n$.  By
  \thmref{karp} with \corref{levelreduction}, $\prob{\Phi(I_{r+5}) >
    (3/4)^r n} < 1/2$.  Observe that $\Phi(I_{r+5})$ is at least the number of
  bridge nodes that are still active in round $r+5$, and each node on an
  active subpath is a bridge node.  Thus, by \lemref{tree}, running
  the algorithm to level $r+5$ is enough to yield a
  shortcutted path
  length of at most $O(2^r) + \Phi(I_{r+5}) \leq O(2^r) + (3/4)^r n$
  with probability at least $1/2$.    Setting both terms equal and
  solving for $r$ gives $r=\log_{8/3}n$.  Thus, with probability at
  least $1/2$, the shortcutted path has length $O(2^r) =
  O(2^{\log_{8/3}n}) = O(n^{1/\lg(8/3)})$.
\end{proof}

\subsection{A Tighter Path-Length Bound (\lemref{seqmain})}\seclabel{seqbetter}

This section tightens the path-length bound to $O(n^{2/3}\log^{4/3} n)$,
thereby proving \lemref{seqmain}.  

The main difference versus \secref{seqsimple} is a better potential
function associated with subproblems.  The $3/4$ bound reduction in
the number of path-active vertices, as stated in
\lemref{nodereduction}, is indeed tight in the worst case.  But the
worst case only occurs when the number of ancestors is equal to the
number of descendents.  When there is imbalance between the two, the
reduction is better.  Consider, for example, the extreme that there
are no descendents --- then the number of path active vertices reduces
by $1/2$ according to \lemref{asymmetry}.  

This section uses the following potential function for a subproblem
$s=(G,P)$: 
\[ \phi(s) = \phi_1(s)C_\phi + \phi_2(s) \ , \hspace{.5cm}
\phi_1(s) = \sqrt{(\alpha+\beta)(\delta+\beta)} \ , \hspace{.5cm}
\phi_2(s) = \eta = (\alpha+\beta+\delta) \ , \]
where $C_\phi > 1$ is a parameter to be set later,
$\alpha = \card{\Anc{G}{P}}$, $\beta=\card{\Bridge{G}{P}}$, and
$\delta = \card{\Desc{G}{P}}$.  The main idea of $\phi_1$ is to
capture imbalance by the geometric mean of the number of ancestors and
descendents, but bridges are included in both counts because bridges
can eventually become ancestors or descendents in induced subgraphs.
This imbalance term is more important, and hence weighted by
$C_\phi$.  The second term $\phi_2$ is added to ensure the following
lemma:
\begin{lemma}\lemlabel{phibounds}
  $\eta \leq \phi(s) \leq (C_\phi+1)\eta $, for $\eta =
  \card{\Anc{G}{P}} + \card{\Bridge{G}{P}} + \card{\Desc{G}{P}}$.
\end{lemma}
\begin{proof}
  The first inequality is trivial: $\phi(s) \geq \phi_2(s) = \eta$.
  For the second, $\phi_1(s)$ is maximized when $\beta = \eta$, giving
  a total of $\phi(s) \leq \phi_1(s)C_\phi + \phi_2(s) = \eta C_\phi + \eta$.
\end{proof}

The next step is to prove a bound analogous to \lemref{nodereduction}
but for $\phi_1(s)$.   

\begin{lemma}\lemlabel{phione}
  Consider any path-relevant subproblem $s=(G,P)$.  Let $\alpha'$,
  $\beta'$, and $\delta'$ be random variables denoting the total
  number of path ancestors, bridges, and descendents,
  respectively, in any child path-relevant subproblems.  Let
  $\phi_1' = \sqrt{(\alpha'+\beta')(\delta'+\beta')}$.  Then
  $E[\phi_1'] \leq \phi_1(s) / \sqrt{2}$.
\end{lemma}
\begin{proof}
  First, observe that any pivots $x$ sampled that are not related to
  the path can only reduce $\phi_1$ either by inactivating an
  ancestor, descendent, or bridge, or by converting a bridge to an
  ancestor or descendent.  Thus, the proof focuses on the final choice
  of $x$.  Let $\alpha$, $\beta$, and $\delta$ denote the number of
  ancestors, bridges, and descendents, respectively, of path $P$ just
  prior to the final choice of $x$ at this point.  Thus
  $\sqrt{(\alpha+\beta)(\delta+\beta)} \leq \phi_1(s)$.  Let
  $\eta = \alpha+\beta+\delta$.

  Note that $\beta \geq 1$ because a path-relevant subproblem must
  have a nonempty path and hence at least one bridge.  The implication
  is that all of the divisors below are nonzero.

  The proof focuses on the sums $\alpha'+\beta'$ or $\delta'+\beta'$.
  The value $\alpha'$ can increase by changing bridges to ancestors,
  but the sum $\alpha'+\beta'$ cannot.  Applying \lemref{asymmetry} in
  the current notation, $E[\alpha' + \beta'] \leq \alpha/2 + \beta$.
  
  The remainder of the proof is analogous to proof of
  \lemref{nodereduction}, with $\phi_1'=0$ if $x$ is a bridge.
  \begin{align*}
    E\left[\phi_1'\right]
    &= \left(\frac{\alpha}{\eta}\right) E[\phi_1' | x\in
      \Anc{G}{P}] + \left(\frac{\delta}{\eta}\right) E[\phi_1'
      | x\in\Desc{G}{P}] \\
    &\leq 
      \left(\frac{\alpha}{\eta}\right)
      E\left[\sqrt{(\alpha'+\beta')(\delta+\beta)}\right] +
      \left(\frac{\delta}{\eta}\right)E\left[\sqrt{(\alpha+\beta)(\delta'+\beta')}\right]
    \\
    &\leq 
      \left(\frac{\alpha}{\eta}\right)
      \sqrt{(E[\alpha'+\beta'])(\delta+\beta)} +
      \left(\frac{\delta}{\eta}\right)\sqrt{(\alpha+\beta)(E[\delta'+\beta'])}
      \tag*{(by Jensen's inequality)}
    \\
    &< 
      \left(\frac{\alpha}{\eta}\right)
      \sqrt{(\alpha/2+\beta)(\delta+\beta)} +
      \left(\frac{\delta}{\eta}\right)\sqrt{(\alpha+\beta)(\delta/2+\beta)}  \tag*{(\lemref{asymmetry})}
    \\
    &=     
      \left(\frac{\alpha}{\eta}\right)
        \sqrt{(1/2)(1+\frac{\beta}{\alpha+\beta})(\alpha+\beta)(\delta+\beta)}
        + \left(\frac{\delta}{\eta}\right)\sqrt{(1/2)(1+\frac{\beta}{\delta+\beta})(\alpha+\beta)(\delta+\beta)}\\
    &= \left(\frac{\phi_1(s)}{\sqrt{2}}\right) \, \frac{\alpha\sqrt{1+\frac{\beta}{\alpha+\beta}}
      + \delta\sqrt{1+\frac{\beta}{\delta+\beta}}}{\eta}\\
    & \leq \left(\frac{\phi_1(s)}{\sqrt{2}}\right) \,
      \frac{\alpha\left(1+\frac{\beta}{2(\alpha+\beta)}\right) +
      \delta\left(1+\frac{\beta}{2(\delta+\beta)}\right)}{\eta} 
      \tag*{(because $\sqrt{1+y} \leq 1+y/2$ for $y\geq 0$)}\\
    & \leq \left(\frac{\phi_1(s)}{\sqrt{2}}\right) \, \frac{(\alpha +
      \beta/2) + (\delta+\beta/2)}{\eta}\\
    & = \phi_1(s) / \sqrt{2} \ . \\[-3em]
  \end{align*}
\end{proof}

Before extending the bound to the full function $\phi$, the following
lemma says that partitioning ancestors, bridges, and descendents
arbitrarily across subproblems does not increase the potential.  This
statement is obvious for $\phi_2$, so the lemma focuses on $\phi_1$. 

\begin{lemma}\lemlabel{phisplit}
  Consider any integers
  $\alpha_1,\alpha_2,\beta_1,\beta_2,\delta_1,\delta_2 \geq 0$ with
  such that $\alpha_i + \delta_i > 0 \implies \beta_i > 0$.  Let
  $\alpha = \alpha_1+\alpha_2$, $\beta=\beta_1+\beta_2$, and
  $\delta=\delta_1+\delta_2$.  Then
  $\sqrt{(\alpha_1+\beta_1)(\delta_1+\beta_1)} +
  \sqrt{(\alpha_2+\beta_2)(\delta_2+\beta_2)} \leq
  \sqrt{(\alpha+\beta)(\delta+\beta)}$.
\end{lemma}
\begin{proof}
  If $\beta_1=0$ or $\beta_2=0$, the claim is trivial.  (By
  assumption, $\beta_1=0$ for example implies that $\alpha_1=0$ and
  $\delta_1=0$.)  Suppose instead that neither is zero.

  Let $y = \alpha+\beta$ and $y_i = \alpha_i+\beta_i$.  Similarly let
  $z = \delta+\beta$ and $z_i = \delta_i+\beta_i$.  Both $y>0$ and
  $z>0$ by assumption on $\beta>0$. Let $\epsilon_y = y_1/y$ and
  $\epsilon_z = z_1/z$. The proof focuses on a more general split of
  $y = y_1+y_2$ and $z=z_1+z_2$.  

  It suffices to show that for all
  $0\leq \epsilon_y,\epsilon_z \leq 1$ that
  $\sqrt{\epsilon_y y\cdot \epsilon_z z} + \sqrt{(1-\epsilon_y)y \cdot
    (1-\epsilon_z)z} \leq \sqrt{yz}$,
  or equivalently (dividing both sides by $\sqrt{yz}$) that
  $\sqrt{\epsilon_y\epsilon_z} + \sqrt{(1-\epsilon_y)(1-\epsilon_z)}
  \leq 1$.
  Fix any $\epsilon_y$, treating $\sqrt{\epsilon_y}$ and
  $\sqrt{1-\epsilon_y}$ as constants.  The expression is maximized at
  $\epsilon_z = \epsilon_y$, so
  $\sqrt{\epsilon_y\epsilon_z} + \sqrt{(1-\epsilon_y)(1-\epsilon_z)}
  \leq \epsilon_y + (1-\epsilon_y) = 1$, which completes the proof.
\end{proof}

Finally, the following lemma considers the full potential $\phi$.
\begin{lemma}\lemlabel{tightnode}
  Consider any path-relevant subproblem $s=(G,P)$, and let $s_1$ and
  $s_2$ be random variables denoting its child subproblems in the
  flattened path-relevant tree (with $\phi(s_i) = 0$ if the child does
  not exist).  Then $E[\phi(s_1) + \phi(s_2)] \leq \phi(s)(1/\sqrt{2}+2/\sqrt{C_\phi})$.
\end{lemma}
\begin{proof}
  Let $\alpha$, $\beta$, and $\delta$ denote the number of ancestors,
  bridges, and descendents, respectively, of path $P$ initially.  Let
  $\alpha'$, $\beta'$, and $\delta'$ be random variables denoting the
  total number of path-active ancestors, bridges, and descendents
  after partitioning around~$x$.  Let $\eta = \alpha+\beta+\delta$ and
  $\eta'=\alpha'+\beta'+\delta'$.  Observe that $\eta = \phi_2(s)$ and
  $\eta' = \phi_2(s_1) + \phi_2(s_2)$.  
  
  \lemref{phione} already bounds the impact of the partitioning on
  $\phi_1$, so the goal here is to consider the contribution of
  $\phi_2$ to the total.  There are two cases depending on the degree
  of imbalance between $\alpha$ and $\delta$.  Assume without loss of
  generality that $\alpha \geq \delta$.

  \textbf{Case 1:} $\delta \leq \alpha/C_\phi$. In this case, the imbalance
  causes $\phi_2$ to decrease significantly.  We have
  \begin{align*}
    E[\eta'] &= (\alpha/\eta) E[\eta' | x \in \Anc{G}{P}] +
                 (\delta/\eta) E[\eta' | x \in \Desc{G}{P}] \\
               &\leq (\alpha/\eta)\cdot (E[\alpha'+\beta']+\delta)  + (\delta/\eta) \cdot (\eta)
    \\
               &\leq (\alpha/\eta) \cdot (\alpha/2 + \beta +
                 \delta) + \delta  \tag*{(\lemref{asymmetry})} \\ 
               &\leq \frac{\alpha(\alpha+\beta) + \alpha\beta}{2\eta} +
                 2\delta\\
               &\leq \frac{\eta^2}{2\eta} +
                 \frac{\alpha\beta}{2\eta} + \frac{2\alpha}{C_\phi}
                 \tag*{($\eta \geq \alpha+\beta$)}\\
               &\leq \eta/2 + \frac{((\alpha+\beta)/2)^2}{2\eta} +
                 \frac{2\eta}{C_\phi} \\
               &\leq \eta/2 + \eta/8 + 2\eta/C_\phi \\
               &< \phi_2(s)/\sqrt{2} + 2\phi(s)/C_\phi \ .
\end{align*}
Adding the contribution from $\phi_1$ and $\phi_2$ gives
\begin{align*}
  E[\phi(s_1)+\phi(s_2)] &= E[\phi_1(s_1) + \phi_1(s_2)]\cdot C_\phi +
                          E[\phi_2(s_1) + \phi_2(s_1)] \\
  &\leq E[\sqrt{(\alpha'+\beta')(\delta'+\beta')}]\cdot C_\phi + E[\eta']
    \tag*{(\lemref{phisplit})}\\
  &\leq \phi_1(s)/\sqrt{2}\cdot C_\phi + E[\eta'] \tag*{(\lemref{phione})}\\
  &\leq \phi_1(s)/\sqrt{2}\cdot C_\phi + \phi_2(s)/\sqrt{2} +
    2\phi(s)/C_\phi \tag*{(reduction to $\eta'$ above)} \\
  &= \phi(s)/\sqrt{2} + 2\phi(s)/C_\phi\\
  &= \phi(s)(1/\sqrt{2}+2/C_\phi) \ .
\end{align*}

\textbf{Case 2:} $\delta > \alpha/C_\phi$.  In this case, $\phi_1(s)$
dominates by so much that it does not matter whether $\phi_2(s)$
decreases at all.  Specifically,
$\phi_1(s) = \sqrt{(\alpha+\beta)(\delta+\beta)} >
\sqrt{(\alpha+\beta)(\alpha+\beta)/C_\phi} =
(\alpha+\beta)/\sqrt{C_\phi} \geq (\phi_2(s)/2)/\sqrt{C_\phi}$,
where the last step follows because
$\alpha\geq \delta \implies \alpha+\beta \geq \eta/2$.  We therefore
have $\phi(s) \geq \phi_1(s)C_\phi \geq \phi_2(s)(\sqrt{C_\phi}/2)$,
implying $\phi_2(s) \leq 2\phi(s)/\sqrt{C_\phi}$.  The reduction to
$\phi_2$ is thus irrelevant.  Putting everything together,
$E[\phi(s_1)+\phi(s_2)] \leq \phi_1(s)/\sqrt{2} \cdot C_\phi +
\phi_2(s) \leq \phi(s)/\sqrt{2} + 2\phi(s)/\sqrt{C_\phi} = \phi(s)
(1/\sqrt{2} + 2/\sqrt{C_\phi})$.

The worse of the two cases is the second, yielding $E[\phi'] \leq
\phi(s) (1/\sqrt{2} + 2/\sqrt{C_\phi})$. 
\end{proof}

As before, define $\Phi(I_r) = \sum_{s\in I_r} \phi(s)$, where $I_r$
is the collection of subproblems at level-$r$ in the flattened tree.
Choose $C_\phi = 8\lg^2 n$.   Linearity of expectation yields the
following: 

\begin{corollary}\corlabel{tightlevel}
  Choose $C_\phi = 8\lg^2 n$, where $n$ is the initial number of
  vertices in the input graph $\hat{G}$.  Then given any collection
  $I_{r-1}$ of subproblems,
  $E[\Phi(I_r) | I_{r-1}] \leq \frac{\Phi(I_{r-1})}{\sqrt{2}} (1+1/\lg
  n)$.\qed
\end{corollary}

\paragraph{Proof \lemref{seqmain}.} 
The proof is analogous to \lemref{weakdiam}.  There are two key
differences: the initial potential is higher, at
$\Phi(I_0)\leq
(C_\phi+1) n \leq 9\lg^2 n$
according to \lemref{phibounds}, and the reduction of active vertices
with each round is slightly worse, at $(1/\sqrt{2}) (1+1/\lg n)$ from
\corref{tightlevel}.  Assuming $n\geq 16$ so that
$(1/\sqrt{2})^6(1+1/\lg n)^6<1/2$, \thmref{karp} implies
$\prob{\Phi(I_{r+8}) > (1/\sqrt{2})^r(1+1/\lg n)^r (9 n\lg^2 n)} <
1/2$.
As before, \lemref{phibounds} states that $\Phi(I_{r+8})$ is an upper
bound on the number of active vertices and hence also the total length
of all remaining subpaths. Thus, by \lemref{tree}, running the
algorithm to level $r+8$ is enough to yield a shortcutted path of
length at most
$O(2^r) + \Phi(I_{r+8}) \leq O(2^r) + O((1/\sqrt{2})^r(1+1/\lg n)^r
n\lg^2 n)$
with probability at least $1/2$.  For $r = O(\lg n)$, this reduces to
$O(2^r) + \Phi(I_{r+8}) \leq O(2^r) + O((1/\sqrt{2})^r n\lg^2 n)$.
Choosing $r=(2/3) (\lg n + 2\lg\lg n) + \Theta(1)$ balances the terms
and yields a path of
length $O(n^{2/3}\log^{4/3} n)$.

\begin{algorithm}[t]
\small
\caption{Modified sequential algorithm for shortcutting}\alglabel{seq2}
\SetKwFunction{footwo}{SeqSC2}
\Indm\footwo{$G=(V,E)$}\\
\Indp
  \nl\lIf{the recursion depth is $\lg n$}{\KwRet{$\emptyset$}}
  \nl $S \gets \emptyset$\\
  \nl \While{$V \neq \emptyset$}{
  \nl   select a vertex $x\in V$ uniformly at random\\
  \nl   $\Succ{}{} \gets \Succ{G}{x}$ \\
  \nl   $\Pred{}{} \gets \Pred{G}{x}$ \\
  \nl   $S \gets S \cup \set{(x,v)| v\in \Succ{}{}} \cup
        \set{(u,x)|u\in\Pred{}{}}$ \tcp*{add shortcuts to/from vertices\\
      having paths from/to $x$, resp.\ }
  \nl   $V_B \gets \Succ{}{} \cap \Pred{}{}$ ; \hspace{1em}
        $V_S \gets \Succ{}{} \backslash V_B$ ; \hspace{1em}
        $V_P \gets \Pred{}{} \backslash V_B$ ; \hspace{1em}
        $V_R \gets V \backslash (V_B \cup V_S \cup V_P)$\\ 
  \nl  $S \gets S \cup \text{\footwo{$G[V_S]$}} \cup \text{\footwo{$G[V_P]$}}$\\
  \nl  $G \gets G[V_R]$
  }
  \nl \KwRet{$S$}
\end{algorithm}

\subsection{Runtime and Number of Shortcuts}\seclabel{seqfull}

This section completes the proof of \thmref{seqfull} by analyzing the
running time and number of shortcuts added.  As stated, however, the
running time of \algref{seq1} is not worst case, so it does not meet
the promise of a Monte Carlo algorithm.  

This section instead analyzes \algref{seq2}.  \algref{seq2} is
obtained from \algref{seq1} by replacing one of the recursive calls
(specifically \foo{$G[V_R]$}) with a loop.  There is also a new base
case after $\lg n$ levels of recursion to make the bounds worst case,
where (as always) $n$ here refers to the number of vertices in the
original graph $\hat{G}$. Aside from this one change, \algref{seq1}
and \algref{seq2} are equivalent.

The following lemma indicates that the main path-length lemmas
(\lemreftwo{weakdiam}{seqmain}) still hold even with the truncated
execution.  More precisely, proof of those lemmas only relies on the
execution reaching a depth much less than $\lg n$ in the
flattened path-relevant tree.

\begin{lemma}
  Consider an execution of \algref{seq1} and the corresponding
  flattened path-relevant tree.  When mapped to an execution of
  \algref{seq2} with the same random choices, the first $\lg n-1$
  levels of the flattened tree all have recursion depth $<\lg n$ in \algref{seq2}.
\end{lemma}
\begin{proof}
  The flattened tree only merges some of the calls corresponding to
  $G[V_R]$.  \algref{seq2} merges all such nodes, which can only
  reduce the depth of nodes further.
\end{proof}

The next lemmas bound the number of shortcuts and running time. 
\begin{lemma}\lemlabel{seqshortcuts}
  Consider a graph $\hat{G}=(\hat{V},\hat{E})$, and let $n =
  \card{\hat{V}}$.  Each execution of \algref{seq2} creates $O(n\log n)$
  shortcuts.
\end{lemma}
\begin{proof}
  Consider a call to \footwo{$G$} on $G=(V,E)$. Each shortcut added
  removes a vertex: if, e.g., $(x,v)$ is created, then either $x \in
  V_B$ or $x \in V_S$, both of which sets are removed from $G$ at the end
  of the iteration.  Thus, there can be at most $\card{V}$ arcs
  added.  

  There are potentially many recursive subproblems, but by the
  same argument they are all disjoint subgraphs.  Thus, the total
  number of arcs added at each level of recursion is $O(n)$.  There
  are $O(\lg n)$ levels by construction, which completes the proof.  
\end{proof}

\begin{lemma}\lemlabel{seqruntime}
  Consider a graph $\hat{G}=(\hat{V},\hat{E})$, and let
  $n=\card{\hat{V}}$ and $m=\card{\hat{E}}$.  \algref{seq2} can be
  implemented to run in $O(m\log n)$ time.
\end{lemma}
\begin{proof}
  Proof is similar to \lemref{seqshortcuts}, getting $O(m)$ total
  time at each level of recursion, assuming that the call \footwo{$G$}
  can be made to run in $O(\card{V}+\card{E})$ time.  Given a sample
  $x\in V$, it is straightforward to implement each search, and build
  the induced subgraphs, to run in time $O(a)$ where $a$ is the number
  of arcs explored.  Each arc is only explored by one search in each
  direction, so the total number of arcs visited is $O(\card{E})$.
  Finally, sampling vertices can be achieved by randomly permuting the
  vertices up front, iterating over that list, and checking whether
  the vertex has already been visited by a search.  This takes a total
  of $O(\card{V})$ time
\end{proof}

\paragraph{Proof of \thmref{seqfull}.}  The full algorithm consists of
$\Theta(\log n)$ independent runs of \algref{seq2}.  For each related
pair $u\preceq v$, each run has probability $\geq 1/2$ of reducing the
distance between those vertices to $O(n^{2/3}\log^{4/3}n)$ by
\lemref{seqmain}. Thus, a Chernoff bound across $\Theta(\log n)$ runs gives a
high-probability result, i.e., failure probability at most $1/n^{c+2}$
for any constant~$c$.  There are only $n^2$ related pairs of
vertices, so a union bound across all of them gives a failure 
probability of at most $1/n^c$.  If there are no failures,
then the stated diameter is achieved.   The running time and number of
shortcuts are obtained by multiplying the bounds from
\lemreftwo{seqruntime}{seqshortcuts} by the $\Theta(\log n)$ runs.  \qed

\section{An Algorithm with Distance-Limited  Searches}\seclabel{parallel}
\newcommand{\dmax}{\accentset{\smallfrown\vspace{-.4ex}}{d}}
\newcommand{\dmin}{\bar{d}}
\newcommand{\malpha}{\bar{\alpha}}
\newcommand{\mbeta}{\bar{\beta}}
\newcommand{\mdelta}{\bar{\delta}}
\newcommand{\meta}{\bar{\eta}}
\newcommand{\mphi}{\bar{\phi}}

This section presents a modified algorithm that is more amenable to
being parallelized.  For now, this algorithm can be viewed as a
sequential algorithm --- discussion of the parallel implementation is
deferred to \secref{details}.  But the main ideas are guided by
certain sequential bottlenecks.  As in \secref{sequential},
$\hat{G} = (\hat{V},\hat{E})$ and $n=\card{\hat{V}}$ are used only to
refer to the original graph.

There are two main obstacles to parallelizing \algref{seq2}, but the
first is more serious.  Finding the set $\Pred{G}{x}$ or $\Succ{G}{x}$
entails a graph search, which can have linear span in a high-diameter
graph.  The solution for this problem is to modify the algorithm to
use a \defn{$D$-limited BFS}, returning only the vertices within $D$
hops of the source $x$, but doing so introduces some other
difficulties.  This section thus focuses on modifying the algorithm to
work with distance-limited searches for appropriate distance~$D$.  

The second obstacle is best exhibited by the loop in \algref{seq2}.
If there are no arcs in the graph, for example, the loop requires
$\Omega(n)$ iterations.  The solution is to perform multiple pivots in
parallel, but in a controlled way that does not sacrifice much
performance.  

The full algorithm is given in pseudocode as \algref{par}.
\secref{limited} walks through the ideas incrementally, guided by
rough intuitions behind the analysis.  The key performance lemma,
analogous to \lemref{seqmain}, is the following:

\begin{lemma}\lemlabel{parmain}
  Let $\hat{G}=(\hat{V},\hat{E})$ be a directed graph, let
  $n=\card{\hat{V}}$, let $m=\card{\hat{E}}$, and assume without loss
  of generality that $m\geq n/2$.  

  Consider any directed path $\hat{P}$ from $u$ to $v$ with
  $\length{\hat{P}} \leq D$, for $D=\Theta(n^{2/3}\log^{4/3}n)$.  Let
  $S$ be the shortcuts produced by an execution of \algref{par} on
  $\hat{G}$ starting with $h=\lg n$.  Then with probability at least
  $1/2$: (1) there exists a path from $u$ to $v$ in
  $G_S = (\hat{V},\hat{E}\cup S)$ with length at most $D/2$, (2) the
  number of shortcuts produced is $\card{S} = O(n\log^2 n)$, and (3) the
  total number of vertices and arcs visited by searches is
  $O(m\log^2 n)$.  Moreover, the maximum distance used for any search is
  $O(n^{2/3}\log^{14}n)$;
\end{lemma}

Using multiple runs of \algref{par} (see \secref{diamalg}) yields the
following:

\begin{theorem}\thmlabel{parfull}
  There exists a randomized algorithm that takes as input a directed
  graph $\hat{G}=(\hat{V},\hat{E})$ and uses distance-limited searches
  with the following guarantees.  Let $n=\card{\hat{V}}$,
  $m=\card{\hat{E}}$, and without loss of generality $m\geq n/2$.
  Then (1) the maximum distance used for any search is
  $O(n^{2/3}\log^{14}n)$; (2) the algorithm produces a
  size-$O(n\log^4n)$ set $S^*$ of shortcuts; (3) the total number of
  vertices and arcs visited by searches is $O(m\log^4n + n\log^8n)$,
  and the searches dominate the overall number of primitive operations
  performed; and (4) with high probability, the diameter of $G_{S^*}$
  is $O(n^{2/3}\log^{4/3}n)$,
\end{theorem}

\paragraph{Updated notation.} If there exists a path with length at
most $d$ from $u$ to $v$, then $u\preceq_d v$.  If $u \preceq_d v$ or
$v\preceq_d u$, then $u$ and $v$ are \defn{$d$-related}.  All of the
notations and definitions in \secref{prelim} that depend on $\preceq$
(i.e., successors, predecesors, ancestors, descendents, bridges) are
augmented with the term ``$d$-limited'' and a subscript $d$ to
indicate that the $\preceq$ in the definition should be replaced by
$\preceq_d$.  For example, $\SuccD{G}{x}{d} = \set{v | x\preceq_d v}$
denotes the \defn{$d$-limited successors} of $x$.  

\begin{algorithm}[t]
\small
\caption{Shortcutting algorithm with distance-limited searches.}\alglabel{par}
\SetKwFunction{foopar}{ParSC}
\Indm\foopar{$G=(V,E)$, $h$}\\
\Indp
  \tcc{The value $h$ indicates how many more levels of recursion to
    perform. $\epsilon_\pi$, $N_k$, $N_L$, and $D$  are all global parameters
    (independent of subproblem)  set later.}
  \nl\lIf{$h=0$}{\KwRet{$\emptyset$}}
  \nl $S \gets \emptyset$\\
  \nl randomly permute $V$, giving vertex sequence
  $X=x_1,x_2,\ldots,x_{\card{V}}$. Mark  each $x_j$ live\\
  \nl split $X$ into subsequences $X_1,X_1,\ldots,X_{2k}$, with
  $\card{X_i} = \card{X_{k-i+1}}=\floor{(1+\epsilon_\pi)^i}$ for $i<k$\\
  \hspace{2.7in}and $\card{X_k} = \card{X_{k+1}} \leq \floor{(1+\epsilon_\pi)^k}$\\
  \nl \For{$i \gets 1$ \KwTo $2k$}{
  \nl   choose random $d \in \set{1,2,\ldots N_L-1}$\\
  \nl   $d \gets d + hN_kN_L - iN_L$ \tcp*{add the distance offset}
  \nl   \ForEach{live $x_j \in X_i$}{
  \nl     $R^-_j \gets \PredD{G}{x_j}{dD}$ ; \hspace{2cm}     $R^+_j
  \gets \SuccD{G}{x_j}{dD}$ \tcp*{core vertices}
  \nl     $F^-_j \gets \PredD{G}{x_j}{(d+1)D} \backslash R_j^-$ ; \hspace{.7cm} $F^+_j \gets
  \SuccD{G}{x_j}{(d+1)D} \backslash R_j^+$ \tcp*{fringe vertices}
  \nl     $S \gets S \cup \set{(x_j,v)| v\in R_j^+\cup F_j^+} \cup
          \set{(u,x_j)|u\in R_j^- \cup F_j^-}$ \tcp*{add shortcuts}
  \nl     append a tag of $j$ to all vertices in $R^+_j \cup R^-_j$
}
  \nl   \ForEach{live $x_j \in X_i$}{
  \nl     remove from $R_j^+$, $R_j^-$, $F_j^+$, and $F_j^-$ vertices with a tag $<
  j$  \tcp*{first core search wins}
  \nl     $V_{B,j} \gets R_j^+ \cap R_j^-$ ; \hspace{2em}
          $V_{S,j} \gets R_j^+ \backslash V_{B,j}$ ; \hspace{2em}
          $V_{P,j} \gets R_j^- \backslash V_{B,j}$ \\
  \nl     $S \gets S \cup \text{\foopar{$G[V_{S,j} \cup F_j^+], h-1$}} \cup
  \text{\foopar{$G[V_{P,j}\cup F_j^-], h-1$}}$ \tcp*{include fringe}}
  \nl  mark all vertices in $\bigcup_j (R^+_j \cup R_j^-)$ as
  dead in $X$\\
  \nl  $V_R \gets V \backslash \bigcup_j (R^+_j \cup R_j^-)$\\
  \nl  $G \gets G[V_R]$
  }
  \nl \KwRet{$S$}
\end{algorithm}

\subsection{The Algorithm}\seclabel{limited}

The main goal is to replace the searches $\Succ{G}{x}$ in
\algreftwo{seq1}{seq2} with $D$-limited searches, for
$D=\Otilde{n^{2/3}}$.  The good news is that \lemref{asymmetry} still
holds when restricted to pivots drawn from $D$-limited
ancestors.  The bad news is that \lemref{rewrite} does not hold.  For
concreteness, consider a path $\ang{v_0,v_1,\ldots,v_\ell}$.  It is
possible, for example, that $x\preceq_D v_{2k}$ but
$x\not\preceq_D v_{2k+1}$ for all~$0 \leq k < \ell/2$.  Thus all the
even vertices would be in $V_S$ and all the odd ones would be in
$V_R$, splitting the path into $\Theta(\ell)$ pieces with no potential
to shortcut them later.  In contrast, when the search is \emph{not}
$D$-limited, $x\preceq v_i$ implies $x\preceq v_k$ for all $k\geq i$,
so $V_S$ contains a single contiguous subpath.
 
The solution is to extend the search a little further and duplicate
vertices.  That is, start with a distance of $dD$, for some
$d=\Otilde{1}$.  Any vertices reached this way are called \defn{core
  vertices}, and they are treated similarly to reached vertices in
\algref{seq1}.  Then extend the search a little farther: to a distance
of $(d+1)D$.  Vertices discovered in the extended search are called
\defn{fringe vertices}, denoted by $F^+$ and $F^-$ in the code.
Fringe vertices $F^+$ and incident arcs are duplicated (similarly for
$F^-$), belonging to both $G[V_R]$ and $G[V_S\cup F^+]$.

The addition of fringe vertices fixes the path-splitting problem,
giving an analog of \lemref{rewrite}, at least for paths of length
$\ell \leq D$.  Consider again the bad example where
$x \preceq_{dD} v_{2k}$ but $x\not\preceq_{dD} v_{2k+1}$.  All of the
even vertices are core vertices, but now all of the odd vertices are
fringe vertices.  Thus, the entire path is indeed contained in the
subgraph $G[V_S \cup F^+]$.  We still reason about the path being
split across subproblems, but fringe vertices on the path can be
treated as belonging to whichever subproblem is better.

Unfortunately, duplicating fringe vertices introduces another problem
--- path-related fringe vertices can be active in multiple
subproblems, thereby destroying the progress bound on $\Phi$.  In the
worst case, almost all of the active vertices could be fringe
vertices, and the total number of active vertices could thus increase
drastically after partitioning around pivot~$x$.

The solution is to select $d$ (for search distance $dD$) randomly from
the range $d \in \set{1,2,\ldots,N_L-1}$, for some $N_L = \Otilde{1}$
to be chosen later.\footnote{Read $N_L$ as ``number of layers''.}  Any
vertices in the fringe for distance $dD$ are in the core for
distances $d'D$, $d' > d$.  Thus, on average, only an
$O(1/N_L)$ fraction of vertices are on the fringe.  For large enough
$N_L$, the addition of these fringe vertices does not impact $\phi(s)$
much.

It is also important that the distances searched never increases.
This is because any progress towards the number of active vertices is
with respect to a particular search distance~$dD$.  The algorithm
therefore selects from $N_L$ distance options, but offset by some
value to reflect future decreases.  With each choice of pivot(s), the
offset decreases by at least $N_L$.  More precisely, as in
\algref{seq2}, the main subroutine consists of a sequence of
iterations, where some pivots are chosen in each iteration.  The total
number of iterations is bounded by some value $N_k$, meaning that the
full range of distances effectively owned by a single call has size
$N_kN_L$. Each time the recursion depth increases, the offset
decreases accordingly by $N_kN_L$.  We thus use a starting offset of
$hN_kN_L$, where $h$ is the number of levels of recursion to perform.
As long as $h=\Otilde{1}$, $N_k=\Otilde{1}$, and $N_L=\Otilde{1}$, the
maximum distance searched is $\Otilde{D} = \Otilde{n^{2/3}}$ as
desired.

\subsubsection*{Searches from Multiple Pivots}

In addition to being more parallelizable, searching from multiple
pivots is also necessary to keep
$N_k = \Theta(\log_{1+\epsilon_\pi} n)=\Theta(\log n / \epsilon_\pi)$
low, where $0<\epsilon_\pi\leq 1$ is chosen later.

A single recursive call \foopar consists of a sequence of iterations,
like \algref{seq2} flattening the recursion of $G[V_R]$.  Each
iteration proceeds as follows.  First, sample a set $\set{x_j}$ of
pivots and perform independent searches from each of them, determining
both the $dD$-limited core ($R_j^+$ and $R^-_j$) and the
$(d+1)D$-limited fringe ($F_j^+$ and $F_j^-$) of each pivot.  Add
shortcuts to and from all reached vertices.  To roughly simulate the
effect of selecting one pivot at a time, if a vertex is part of
$x_j$'s core, then it is removed from the core and fringe sets for any
$x_j'$ with $j'>j$.  Next, calculate the sets $V_{S,j}$ and $V_{P,j}$
as in \algreftwo{seq1}{seq2} and launch the recursive subproblems
$G[V_{P,j} \cup F_j^-]$ and $G[V_{S,j} \cup F_j^+]$.  Finally, remove
all core vertices from the graph and start the next iteration.

\algref{par} uses the following process to control the pivot sampling.
Randomly permute all of the vertices at the start of the call,
creating a sequence $x_1,x_2,\ldots$ of pivots to consider.  All
pivots are initially \emph{live}; the live pivots are those still in
the graph.  In each iteration, select the next group of pivots from
the sequence, where the size of the group is discussed below.  Perform
searches from each live pivot, and ignore the dead ones.  When a core
vertex is removed from the graph, the vertex is also marked dead in
the pivot sequence.

\paragraph{Number of pivots.} 
The number of pivots (live or dead) selected in each iteration is
controlled by the parameter $0 < \epsilon_\pi \leq 1$.  For the first
$\Theta(1/\epsilon_\pi)$ iterations, only one pivot is used.  In
subsequent iterations, the number of pivots increases geometrically by
roughly $(1+\epsilon_\pi)$.  Were the only goal to keep the number of
times a vertex is reached in a search to $O(\log n)$, setting
$\epsilon_\pi=1$ and following the geometric increase would be
sufficient.  To bound the number of times a path can split in a single
iteration, however, it is important to achieve a tighter bound.  There
are $2k$ iterations total, where $k$ is chosen to be large enough to
include all vertices according to the following group sizes.  The
first $k$ iterations follow a geometric increase, and the next $k$
iterations follow a geometric decreases.  More precisely, the number
of pivots considered in both iteration $i$ and $2k-i+1$ is
$\floor{(1+\epsilon_\pi)^i}$, but iterations $2k$ and $2k+1$ can be
smaller.

\begin{algorithm}[t]
\small
\caption{Diameter reduction with distance-limited searches.}\alglabel{pardiam}
\SetKwFunction{pardiam}{ParDiam}
\Indm\pardiam{$\hat{G}=(\hat{V},\hat{E})$}\\\Indp
\nl \For{$i\gets 1$ \KwTo $\Theta(\log n)$}{
\nl   \ForEach{$j \in \set{1,2,\ldots,\Theta(\log n)}$} {
\nl      $S_j \gets \text{\foopar{$G',\lg n$}}$, aborting if number
of shortcuts or work exceeds \lemref{parmain}}
\nl   $\hat{E} \gets \hat{E} \cup \left(\bigcup_j S_j\right)$ \tcp*{add more arcs to $\hat{G}$}
    }
\nl \KwRet $\hat{G}$
\end{algorithm}

\subsection{Notation and Shorthand}

It is often convenient to refer to iterations of the loop in
\algref{par}.  During iteration $i$, quite a bit happens: some pivots
are processed, some searches are performed, some induced subgraphs are
built, etc, and the claims throughout refer to those objects.
Defining every term concretely in every lemma statement or proof gets
tedious and unwieldy. Instead, this paper adopts some notational
conventions consistent with the pseudocode in \algref{par}, using the
variables to implicitly adopt the meaning of the code.

Concretely, for iteration~$i$ on graph $G=(V,E)$, the following
notations are used with the same meaning as the pseudocode: $h$,
$X_i$ meaining the pivot sequence, and $d$ meaning the random distance
chosen.  Moreover, for each $x_j\in X_i$, whenever notations $R_j^+$, $R_j^-$,
$F_j^+$, $F_j^-$, $V_{S,j}$, or $V_{P,j}$ appear, they should also be
interpreted to have the meaning laid out in the pseudocode. 

\paragraph{Min and max distances.}
In each iteration~$i$, the algorithm chooses a random distances in
some size-$(N_L-1)$ range, but at an offset that depends on the
iteration.  The values $\dmin{}$ and $\dmax{}$ denote the bounds of
the range, i.e., drawing random
$d \in \set{\dmin{}, \dmin{}+1,\ldots,\dmax{}-1}$.  Here
$\dmin{} = hN_kN_L-iN_L+1$, and $\dmax{} = hN_kN_L-(i-1)N_L$.  The
minimum possible search distance for a core search is $\dmin{}D$.  The
maximum possible search distance for a fringe search is $\dmax{}D$.
Note that $\dmin{}$ and $\dmax{}$ both rely on the current iteration
$i$ and recursion height $h$, which is consistent with the general
notational shorthand.  These min and max distances are useful for
classifying vertex relationships as follows:

\begin{definition}\deflabel{dist}
  Consider any iteration $i$ of \algref{par}.
  \begin{closeitemize}
  \item Vertices $u$ and $v$ are \defn{never related} if
    $u\not\preceq_{\dmax D} v$ and $v \not\preceq_{\dmax D} u$.  
  \item Vertices $u$ and $v$ are \defn{partly related} if
    $u\preceq_{\dmax D} v$ or $v\preceq_{\dmax D} u$. 
  \item Vertices $u$ and $v$ are \defn{fully related} if $u
    \preceq_{\dmin D} v$ or $v\preceq_{\dmin D} u$.   If $u$ and $v$
    are fully related, then they are also partly related.
  \end{closeitemize}
  When comparing a vertex $v$ and a path $P$, the same terms apply in
  the natural way.  For example, if $v$ is fully related with any
  vertex in $P$, then $v$ and $P$ are fully related.
\end{definition}

\subsection{Full Diameter-Reduction Algorithm and Proof
  of~\thmref{parfull}}\seclabel{diamalg}

Like the algorithm in \secref{sequential}, to achieve diameter
reduction with high probability requires multiple passes of
\algref{par}.  But now more passes are necessary. The full algorithm,
shown in \algref{pardiam}, is as follows.  Perform $\Theta(\log n)$
iterations.  In each iteration, perform $\Theta(\log n)$ independent
executions of \algref{par} on the current graph.  Add to the graph all
of the shortcuts produced thus far, and continue to the next iteration
on the updated graph. 

The main reason for the extra passes of \algref{par} is that, due to
the $\Otilde{D}$-limited searches, the analysis only considers paths
of length $D$.  The distance $D$ is chosen to be large enough so that
each iteration is enough to reduce the length of the path to $D/2$,
with high probability, but a longer path needs to be subdivided.

\vspace{1em}
\noindent\emph{Proof of \thmref{parfull}, assuming \lemref{parmain}}. Consider any two vertices
$u\prec v \in V$.  Let $\Delta_i$ denote the length of the shortest
path from $u$ to $v$ in the graph after iteration $i$ of the outer
loop of \algref{pardiam}.  The main claim is that with high
probability, $\Delta_i \leq D \cdot \max{n/(D2^i), 1}$.  Thus, when
\algref{pardiam} returns, the diameter bound is met. 

The proof is by induction on $i$.  For $i=0$, the length of the
shortest path is at most $n$, so $\Delta_0 \leq n =D \cdot n/(D2^0)$.
For the inductive step (going from iteration $i$ to $i+1$), consider
the shortest path $P$ from $u$ to $v$ in the current graph. If
$\length{P} \leq D$, then the path is already short enough.
Otherwise, subdivide the path into at most $(n/(D2^i))$ subpaths, each
of subpaths, each of length at most $D$.  Consider each subpath.  By
\lemref{parmain}, a single execution of \algref{par} shortens the
subpath's length to $D/2$ with constant probability.  Thus, using a
Chernoff bound, $\Theta(\log n)$ runs shorten the subpath to $D/2$
with high probability.  Taking a union bound over all subpaths gives
high probability that all subpaths are shortened.  Concatenating the
subpaths yields a path of length $(D/2)\cdot
n/(D2^{i})=n/(D2^{i+1})$.  

The search distance follows directly from \lemref{parmain}.  The
number of shortcuts follows from \lemref{parmain} by multiplying by
the number of $\Theta(\log^2 n)$ runs.  As for the bound on total
number of arcs visited, observe that the graph size is at most
$\hat{E} + O(n\log^4n)$ at the end.  Thus, by \lemref{parmain}, each
run of \algref{par} visits $O( (m\log^4n)\log^2 n) = O(m\log^6n)$
arcs.   Multiplying by $\Theta(\lg n)$ runs completes the proof.\qed

\subsection{Bounds on Number of Vertices Searched}

This section bounds the number of times each vertex can be searched in
each iteration of the main loop in \algref{par}.  The main lemma,
stated next, is used to prove two corollaries.  The first corollary
says that, with high probability, a vertex is not searched more than
$O(\log n)$ times, where $n=\card{\hat{V}}$.  The second corollary
gives a tighter bound, but only in expectation.

\begin{lemma}\lemlabel{numpivots}
  Consider any iteration $i$ of the loop in the call
  \foopar{$G=(V,E),h$}.  Let $y = \sum_{i'=1}^{i-1} \card{X_i}$ be the
  number of pivots processed before iteration $i$ begins.  Let
  $G_i = (V_i,E_i)$ be the remaining subgraph at the start of the
  iteration and let $\tau = \card{V}\ln{n}/y$ for any
  $n\geq 2$.\footnote{The $n$ here can be anything, but the intent is
    $n=\card{\hat{V}}$, the number of vertices in the full graph.}
  Then for every $v\in V$ and constant $c\geq 2$:

  With probability at least $1-1/n^{c-1}$: either
  $v\not\in V_i$, or $\card{\PredD{G_i}{v}{\dmax{}D}} \leq c\tau$ and
  $\card{\SuccD{G_i}{v}{\dmax{}D}} \leq c\tau$.
\end{lemma}
\begin{proof}
  All searches before iteration~$i$ have distance larger than $\dmax{}D$.
  Moreover, arcs are not added to the graph on each iteration, so the
  number of $\dmax{}D$-limited predecessors of $v$ can only decrease or stay
  the same with each iteration.  Thus, if $v$ is to end with $>c\tau$
  live predecessors, it must have $>c\tau$ live predecessors the
  entire time.  The remainder of the proof bounds the probability of
  that event occurring.  The argument for successors is symmetric.

  For $y < c\ln n$, the claim is vacuous, so consider instead that
  $y \geq c\ln n$.  Let $x_1,x_2,\ldots,x_y$ denote the sequence of
  pivots chosen before iteration $i$ begins.  While the number of
  $\dmax{}D$-limited predecessors is above threshold, there are at most
  $\card{V}$ choices of $x_j$ so we have $\prob{\text{$x_j$ live and }
    x_j \preceq_{\dmax D} v} \geq c\tau/\card{V} = c\ln n / y$.
  Thus, for $y \geq c\ln n$, $\prod_{j=1}^y \prob{\text{$x_j$ dead or
    } x_j \not\preceq_{\dmax{}D} v} \leq (1-c\ln n/y)^y \leq 1/e^{c\ln n} =
  1/n^c$.  Taking the union bound across the two failure events (predecessors
  and successors) gives failure probability $2/n^c < 1/n^{c-1}$ for
  $n\geq 2$. 
\end{proof}

\begin{corollary}\corlabel{numvisits}
  Choose any $\epsilon_\pi \leq 1$.  Consider any iteration $i$ and
  let $X_i$ be the random set of pivots selected.  Then with high
  probability with respect to $n$, no vertex is $\dmax{} D$-related to
  more than $O(\log n)$ live pivots in~$X_i$.
\end{corollary}
\begin{proof}
  Use $G=(V,E)$ to refer to the graph at the beginning of the call,
  before the first iteration of the loop.  Let $y$ be the number of
  pivots considered before the iteration in question, and let $v$ be a
  vertex to analyze.  By choice of $\epsilon_\pi$,
  $\card{X_i} \leq 3y$.  (It can only be this large due to roundoff.)
  
  \lemref{numpivots} has failure probability $1/n^{c-1}$, where
  we can choose whatever constant $c$ we want.  We can add the failure
  probabilities at the end by a union bound, so suppose for now that
  \lemref{numpivots} applies.  Then $v$ has $O(\card{V}\log n / y)$
  $\dmax{}D$-limited predecessors and successors.  

  There are two cases.    If $y > \card{V}/8$, then
  $O(\card{V}\log n /y ) = O(\log n)$.   Searching from \emph{every}
  remaining vertex would thus only result in $O(\log n)$ searches
  reaching $v$.  To complete the proof, take a union bound across all
  $v$ to get at most a $n^{c-2}$ failure probability. 

  If $y \leq \card{V}/8$, then $y + \card{X_i} \leq \card{V}/2$.  For
  every pivot position $x_j \in X_i$, there are thus at least
  $\card{V}/2$ options to draw from.  So we have $\prob{\text{$x_j$
      live and } x_j \preceq_{dD} v} = O(\card{V}\log n / (y\card{V}))
  = O(\log n / y) = O(\log n / \card{X_i})$,
  where the last step follows from $\card{X_i} \leq 3y$.  Applying a
  Chernoff bound across all $\card{X_i}$ pivots, the number of
  searches that reach $v$ is $O(\log n)$ with high probability.  To
  complete the proof, take a union bound across all $v$, adding the
  failure probabilities from the Chernoff bound and \lemref{numpivots}.
\end{proof}

\begin{corollary}\corlabel{pivotbound}
  Consider any iteration $i$ and vertex $v$, and suppose that
  $\epsilon_\pi \geq 1/n$ for $n=\card{\hat{V}}$.  As long as 
  $\card{X_i} > 1$, then the expected number of times that vertex $v$
  is visited by searches is $O(\epsilon_\pi \log n)$.
\end{corollary}
\begin{proof}
  Let $G=(V,E)$ denote the graph before the first iteration of the
  loop.  Suppose that \lemref{numpivots} applies, so $v$ has
  $O(\card{V}\log n/y)$ $dD$-limited predecessors and successors,
  where $y$ is the number of pivots processed before this iteration.
  The failure event can only increase the expectation by an additive
  $\prob{\text{failure}} \cdot \card{V} \leq (1/n^c) n = 1/n =
  O(\epsilon_\pi)$.

  The main observation is that when $\card{X_i} > 1$, $i$ is not in
  one of the first or last $\Theta(1/\epsilon_\pi)$ iterations, and
  the number of vertices $\card{X_i}$ processed in this iteration is
  at most an $O(\epsilon_\pi)$-fraction of the number of remaining
  vertices.  Similarly, $\card{X_i} = O(\epsilon_\pi y)$.

  If $y > \card{V}/8$, then $O(\card{V}\log n/y) = O(\log n)$, meaning
  that $v$ has only $O(\log n)$ $dD$-limited predecessors or ancestors
  remaining.  This iteration thus processes only $O(\epsilon_\pi \log
  n)$ of them in expectation.

  If $y \leq \card{V}/8$, let $Y_j$ be an indicator random variable
  for the event that $x_j \preceq_{dD} v$.  Since there are
  $\Omega(\card{V})$ pivots to choose from,
  $\prob{Y_j} = O(\card{V}\log n / (y\card{V})) = O(\log n / y)$.  It
  follows that
  $E[\sum_{x_j\in X_i} Y_j] = \card{X_i} \cdot O(\log n / y) =
  O(\epsilon_\pi \log n)$.
\end{proof}

The next lemma bounds the number of fringe vertices and arcs explored
in a single iteration~$i$.  Note that a particular vertex may be a
fringe vertex for multiple searches.  The lemma counts the total
number of times that each vertex appears on the fringe.  Similarly,
arcs may be explored by multiple fringe searches instance of the
vertex, but just once per search.  (An arc $(u,v)$ is explored by
$x_j$'s fringe search if either $u$ or $v$ is a fringe vertex.)  

\begin{lemma}\lemlabel{fringe}
  Consider an iteration $i$ and $N_L\geq 2$.  Let $V'$ be any subset
  of the vertices remaining in the graph, and let $E'$ be any subset
  of the arcs remaining in the graph.  Then the total number of fringe
  vertices also in $V'$ is $O(\card{V'}\log n/N_L)$ in expectation
  over choice of distance.  Similarly, the total number of arcs
  explored by fringe searches is also in $E'$ is
  $O(\card{E'}\log n/N_L)$ in expectation.
\end{lemma}
\begin{proof}
  By \corref{numvisits}, with high probability each vertex in $V'$ is
  visited by at most $O(\log n)$ searches (fringe or core).  A failure
  event can only increase the expectation by an additive
  $\prob{\text{failure}} \cdot \card{\hat{V}}^2 \leq (1/n^c) n^2 \ll
  \log n$
  for appropriate choice of constant $c$ in the high-probability
  bound.  Thus the remainder of the proof assumes that each vertex is
  not visited too many times.

  Fix any arbitrary set of pivots that satisfies \corref{numvisits}.
  Consider any $v \in V'$ and pivot $x_j$.  Let $Y_j^v$ be an
  indicator random variable for the event that $v$ is on $x_j$'s
  fringe, where $Y_j^v = 0$ if $x_j$ and $v$ are never related.  For a
  partly related $x_j$, $v$ is only on the fringe at one distance, so
  $\prob{Y_j^v} \leq 1/(N_L-1) \leq 2/N_L$.  The total number of times
  $v$ is on the fringe is thus
  $E[\sum_{x_j \in X_i} Y_j^v] = O(\log n) \cdot (2/N_L) = O(\log
  n/N_L)$.  Summing across all $v$ gives $E[\sum_{v \in V'} \sum_{x_j
    \in X_i} Y_j^v] = O(\card{V'} \log n / N_L)$.

  The same argument applies to arcs, observing that the arc is
  explored whenever its endpoints are at the right distance. 
\end{proof}

\subsection{Setting Up the Path-Relevant Tree}

The definition of path-relevant subproblems and the path-relevant tree
differ slightly from \secref{seqover} to account for the key changes.
\algref{seq2} looks closer to \algref{par} than \algref{seq1} does, so
it is worth contrasting the changes with 
\algref{seq2}.  

Nodes in the (unflattened) path-relevant tree are analogous to those
in \secref{seqover}.  Each node corresponds to an iteration of the
for loop in \algref{par}, and associated with the node is the
path-relevant subproblem $(G,P)$ at the start of the iteration.
Analogously, each node in \secref{seqover} corresponds to an iteration
of while loop in \algref{seq2}.   

The following lemma, analogous to \lemref{rewrite}, considers the
effect of an iteration on the path-relevant subproblems.  Unlike
\lemref{rewrite}, there may be more than two path-relevant subproblems. 

\begin{lemma}\lemlabel{parrewrite}
  Let $P=\ang{v_0,\ldots,v_\ell}$ be a nonempty path in $G=(V,E)$ with
  $\ell \leq D$ and consider the effect of a single iteration of the
  for loop in \algref{par}.  Let $X_i$ be the set of pivots selected
  for this iteration, and let $d$ be the distance chosen for the core
  search.  Then the following are the outcomes:
  \begin{closeenum}
  \item (Base case.) If $X_i$ contains a live $dD$-limited bridge of
    $P$, then the shortcuts $(v_0,x)$ and $(x,v_\ell)$ are created.
    There are no path-relevant subproblems.
  \item If none of the live pivots in $X_i$ are $dD$-related to $P$,
    then $P$ is entirely contained in $G[V_R]$; the one path-relevant
    subproblem is thus $(G[V_R],P)$ --- the next iteration of the for
    loop.  
  \item Suppose that just one live pivot in $x_j \in X_i$ is
    $dD$-related to $P$ but it is not a bridge.  Then there exists a
    2-way splitting $P=P_1\concat P_2$ splitting of path $P$ such that
    either (i) $P_1$ is fully contained in $G[V_R]$ and $P_2$ is fully
    contained in $G[V_{S,j}\cup F_j^+]$, or (ii) $P_1$ is fully
    contained in $G[V_{P,j} \cup F_j^-]$ and $P_2$ is fully contained
    in $G[V_R]$.
  \item Suppose that $r$ of the live pivots in $X_i$ are $dD$-related
    to $P$ but none of them are bridges.  Then there exists an
    $(r+1)$-way splitting
    $P=P_1\concat P_2 \concat \cdots \concat P_{r+1}$ corresponding to
    $r$ consecutive applications, in pivot order, of the 2-way split
    above.  (Some of the paths $P_i$ may be empty.)  It follows that
    $(G,P)$ gives rise to at most $r+1$ path-relevant subproblems.
  \end{closeenum}
\end{lemma}
\begin{proof}
  (Case 1.) Suppose that some $dD$-related bridge $x_j$ is selected.
  Then by definition there exist vertices $v_a,v_b \in P$ such that
  $v_a\preceq_{dD} x_j \preceq_{dD} v_b$.  Since the path has length
  at most $D$, it follows that $v_0 \preceq_{(d+1)D} x_j
  \preceq_{(d+1)D} v_\ell$.  Since $v_0$ and $v_\ell$ are within the
  fringe search distance, the claimed shortcuts are added.

  (Case 2.) None of the vertices are $dD$-related to $P$.   Then none
  of the vertices in $P$ are removed from $V_R$.  

  (Case 3.) If $x_j$ is a $dD$-limited ancestor.  (The case for a
  descendent is symmetric.) Then there exists a vertex $v_j \in P$
  such that $x_j \preceq_{dD} v_j$.  The $v_j$ be the earliest such
  vertex.  Then $x_j \preceq_{(d+1)D} v_{j'}$ for all $j' \geq j$.
  Since $x_j$ is the first (and only) $dD$-related pivot in pivot
  sequence, none of those vertices are in another pivot's core.  So
  $v_j,v_{j+1},\ldots,v_\ell \in R_j^+ \cup F_j^+$.  Moreover, $x_j$
  is not a bridge, so none of them are in $R_j^-$, meaning they are in
  $V_{S,j} \cup F_j^+$.

  (Case 4.) Consider the path-related pivots in permutation order,
  applying the above case inductively.  
 \end{proof}

As in \secref{sequential}, the bulk of the analysis is with respect to
the \emph{flattened} path-relevant tree.  The question is when nodes
should be flattened.  The most natural choice --- flatten when case~2
applies --- turns out not to work.  The problem is that knowing
whether the searches reached the path reveals information about the
distance chosen, which changes the distributed over fringe searches.
Notably, \lemref{fringe} requires the full range of distance choices.  

Instead, an iteration is merged with the next iteration whenever the
pivot $x_j$ is never related with the path $P$.  A node in the
flattened tree thus consists of a sequence of iterations such that
only the last iteration selects a pivot that partly related with the
path.  A partly related vertex pivot may, depending on choice of
distance, result in the path splitting.  But the analysis
pessimistically charges for the split.

The following lemma, analogous to \lemref{tree}, bounds the length of
the path at depth-$r$ in the tree.  Since the degree of each node in
the tree is now a random variable, the bound here holds with failure
probability $\leq 1/8$. 

\begin{lemma}\lemlabel{partree}
  Consider any graph $\hat{G} = (\hat{V},\hat{E})$ and any path
  $\hat{P}$ from $u$ to $v$ with $\length{P} \leq D$.  Consider an
  execution of \algref{par} starting with $h=\lg n$, for
  $n = \card{\hat{V}}$, with parameter $\epsilon_\pi$ satisfying
  $\epsilon_\pi \leq 1/\lg^3n$.  Let $S$ be the shortcuts produced,
  and let $\set{(G_1,P_1),\ldots,(G_k,P_k)}$ denote the set of
  path-relevant subproblems at level $r \leq \lg n$ in the flattened
  path-relevant tree.  Then with probability $\geq 7/8$, there is a
  $u$-to-$v$ in $G_S = (\hat{V}, \hat{E} \cup S)$ of length at most
  $O(2^r) + \sum_{i=1}^k\length{P_i}$.
\end{lemma}
\begin{proof}
  This proof focuses on showing that, with probability at least $7/8$,
  the number of nodes in the flattened path-relevant tree is
  $O(2^r)$.  Otherwise the proof is the same as \lemref{tree}.

  Number the nodes in a particular row as $1,2,\ldots,q$.  Each node
  selects at least one partly path-related pivot.  Let $z_t$ be random
  variables denoting the number of additional partly-path-related
  pivots selected in node~$t$ of the tree.  By \lemref{parrewrite},
  each pivot gives rise to at most one additional child node.  So the
  number of nodes in the next row is
  $\leq \sum_{j=1}^q (2+z_t) = 2q + \sum_{t=1}^q z_t$.  If only one
  pivot is in the pivot set, then $E[z_t]=0$; otherwise,
  \corref{pivotbound} implies
  $E[z_t] = O(\epsilon_\pi \log n) = O(1/\log^2n)$ in general.  By
  Markov's inequality,
  $\prob{\sum_{t=1}^q z_t \geq 8\lg n \cdot E[\sum_{t=1}^q z_t]} \leq
  1/(8\lg n)$.
  Thus, with probability at least $1-1/(8\lg n)$, the number of nodes
  in the next row is
  $\leq 2q + 8\lg n \cdot E[\sum_{t=1}^q z_t] = 2q + 8\lg n\cdot
  O(1/\log^2n) = 2q(1+O(1/\log n))$.
  Taking a union bound across all $r\leq \lg n$ rows, the probability
  that even one row increases by more than $2(1+O(1/\log n))$ factor
  is at most $1/8$.  Assuming no such larger increase occurs, the
  total number of nodes in the tree is
  $\leq \sum_{r'=0}^r 2^{r'}(1+O(1/\log n))^{r'} = \sum_{r'=0}^r
  2^{r'} e^{O(1)} = O(2^{r})$.
\end{proof}

\subsection{Progress on Path-Active Vertices} 

This section is analogous to \secref{seqbetter}.  The goal is to argue
that after $(2/3)\lg n + o(\log n)$ rounds (or levels in the tree),
the number of active vertices drops below~$D/2$ with constant
probability. (Or rather, to select the right choice of $D$ so that
this is true.)  Fortunately, the key lemmas from \secref{seqbetter}
can still be applied, since they only rely on the asymmetry
of~\lemref{asymmetry}.  And although multiple pivots may be selected
in any iteration, the analysis only leverages the progress caused by
the first pivot in sequence order.  

There are some differences to the analysis as well. Most notably, the
additional of fringe nodes increases the number of path-active
vertices.  That increase has an effect on the expected reduction to
the potential, which must be bounded here.  It is possible for the
potential to \emph{increase} when advancing to the next round.
\thmref{karp}, however, requires that $\Phi$ never increase.  This
section updates the definition of $\Phi$ and adds a little extra
machinery.

Due to the distance-limited searches, the definition of path active is
updated with respect to being partly related.  That is, a vertex $v$
is \defn{path active} at some level $r$ in the flattened tree if (1)
$v$ is part of some path-relevant subproblem at level $r$, and (2) $v$
is partly related to the path in that subproblem.

The potential function of a subproblem is the same as in
\secref{seqbetter}, but with the updated definition of being path
active.  Concretely, for subproblem $(G,P)$,
\[ \phi(s) = \phi_1(s)C_\phi + \phi_2(s) \ , \hspace{.5cm}
\phi_1(s) = \sqrt{(\alpha+\beta)(\delta+\beta)} \ , \hspace{.5cm}
\phi_2(s) = \eta = (\alpha+\beta+\delta) \ , \]
where $\alpha = \card{\AncD{G}{P}{\dmax D}}$, $\beta =
\card{\BridgeD{G}{P}{\dmax D}}$, and $\delta =
\card{\DescD{G}{P}{\dmax D}}$.  

\paragraph{Fringe nodes.}  It is useful to track the provenance of
vertices, treating fringe vertices as new vertices.  Specifically, a
path-active vertex $v$ is said to be \defn{preserved} by the execution
of a specific node in the flattened tree if (1) $v$ is still path
active in a child node $s'$, and (2) $v$ was not added to $s'$ due to
a fringe search.  That is, if for example $s'$ operates on the graph
$G[V_{S,j} \cup F^+_j]$, and $v\in F^+_j$, then $v$ is \emph{not}
preserved here. (But the vertex $v$ could be preserved with respect to
a different subproblem.)  Given this view of fringe nodes, the
analysis can adopt the tools from \secref{sequential} to consider the
effects of the core searches.

\paragraph{How to reason about the randomness.}  
There is one key difference in the reasoning.  Think of the randomness
in the following way: first, reveal enough of the randomness in the
pivots just to reveal which pivots are partly path related and fully
path related, but not any more specific than that.  This first step is
enough to determine whether the iteration is the final iteration of
the flattened node.  Then determine the distance searched.  Finally,
resolve the specific pivot choices.  This process has the same
probabilities of any outcome as \algref{par}, but reasoning about the
randomness in this way helps.

For concreteness, here is a restatement of \lemref{asymmetry} in the
new context.  This version only states a progress argument if the
pivot selected is fully related to the path --- the reason is that a
different argument will be applied for partly related pivots.
Specifically, active vertices that are not fully related to the path
are, by definition, not related at any distance less than $\dmin D$.
All such vertices are always inactive in child subproblems, regardless
of random choices.  The only unknown is thus what happens to the
fully-related active vertices.
  
\begin{lemma}\lemlabel{parasymmetry}
  Consider any subproblem $(G,P)$.  Let $A = \AncD{G}{P}{\dmin D}$ be
  the fully-related ancestors of path $P$.  Consider any search
  distance $dD \geq \dmin D$.  Let $x_j$ be the first fully-related
  pivot selected, and suppose that $x_j$ is drawn uniformly at random
  from $A$. Let $\malpha = \card{A}$.  Let $\malpha' \leq \malpha$ denote
  the number of vertices in $A$ that are preserved. (Recall that
  preserved means with respect to core searches only.)  Then
  $E[\malpha' | x \in A] < \malpha/2$.
\end{lemma}
\begin{proof}
  The proof is similar to \lemref{asymmetry}, except that a subset of
  vertices is considered, and all relationships are with respect to
  $\preceq_{dD}$ instead of $\preceq$.  As before, the goal is to show
  that the preserves relation is antisymmetric for all $u,v \in A$. 
  
  It is possible  $u \in A$ be a $dD$-limited bridge.  (It is not a
  bridge at distance $\dmin D$, but it could be a bridge at greater
  distance.)  Bridges do not have any path-relevant subproblems, so
  they do not preserve any other vertices. 

  Consider any pair of vertices $u,v \in A$ such that $u$ preserves
  $v$.  The logic follows proof of \lemref{asymmetry} with the same
  two cases, summarized briefly here.  If $u\preceq_{dD} v$, then $v$
  cannot preserve $u$.  If $u$ and $v$ are $dD$-unrelated, then
  consider the earliest vertices $v_a$ with $u\preceq_{dD} v_a$ and $v_b$
  with $v \preceq_{dD} v_b$ on the path.  For $u$ to preserve $v$, it
  must be that $b < a$, so $v$ cannot also preserve $u$.
\end{proof}

The other lemmas in \secref{seqbetter}, e.g., \lemref{tightnode}
essentially just build algebraically off \lemref{asymmetry}, so the
analogs with respect to \lemref{parasymmetry} hold with the same
proof, but for two issues that require care when applying the lemmas
--- \lemref{parasymmetry} neither copes with fringe nodes nor with
active vertices that are partly related but not fully related.  Since
the algebra remains the same, these lemmas are not reproved.

Much of the complexity that follows arises from the specific choice of
potential function.  Using just the linear function $\phi_2$, as in
\secref{seqsimple}, would simplify many of the details.  But the bound
would be worse.

The following lemma implies that fringe nodes can be cleanly
factored-out of the $\phi$.   The number $f$ of fringe nodes is
included twice in the $\phi_1$ term, which overcharges the
contribution of fringe nodes.

\begin{lemma}\lemlabel{phifringe}
  Consider any $\alpha\geq 0$, $\beta\geq 0$, $\delta \geq 0$, and let
  $\eta = \alpha+\beta+\delta$.  Let
  $\phi_1 = \sqrt{(\alpha+\beta)(\delta+\beta)}$ and
  $\phi = \phi_1\cdot C_\phi+\eta$.  For any $f \geq 0$ and
  $C_\phi\geq 2$,
  \[ \sqrt{(\alpha+\beta+f)(\delta+\beta+f)} \cdot C_\phi + \eta + f
  \leq \phi (1+1/\sqrt{C_\phi}) + 3fC_\phi^2 \ . \] 
\end{lemma}
\begin{proof}
  Let $y=\alpha+\beta$ and $z=\delta+\beta$, and without loss of
  generality suppose that $y \leq z$.  Consider the $\phi_1$ term
  first, and for ease of reference, use superscript $^f$ to refer to the version
  of the terms including $f$, e.g., $\phi_1^f = \sqrt{(y+f)(z+f)}$.
  We have
  $\phi_1^f \leq \sqrt{yz} + \sqrt{yf+zf} + f \leq \phi_1 +\sqrt{2fz}
  + f$.  There are two cases.
  
  Case 1: $2f \leq z/C_\phi^{3/2}$.  Then $\sqrt{2zf} \leq
  \sqrt{z^2/C_\phi^3} = z/C_\phi^{3/2} \leq \eta/C_\phi^{3/2}$.  Putting
  everything together gives $\phi^f = \phi_1^f C_\phi + \eta + f \leq
  (\phi_1 + \eta/C_\phi^{3/2} + f)C_\phi + \eta + f = \phi +
  \eta/\sqrt{C_\phi} + f(C_\phi+1) \leq \phi(1+1/\sqrt{C_\phi}) +
  2fC_\phi$.  

  Case 2: $2f > z/C_\phi^{3/2}$.  Then
  $\sqrt{2fz} \leq \sqrt{(2f)(2fC_\phi^{3/2})} = 2fC_\phi^{3/4} <
  2fC_\phi$.
  Putting everything together gives $\phi^f = \phi_1^f C_\phi + \eta +
  f \leq (\phi_1 + 2fC_\phi + f) C_\phi + \eta + f = \phi + 2fC_\phi^2 +
  fC_\phi + f \leq \phi + 3fC_\phi^2$ for $C_\phi \geq 2$. 

  Taking the larger of the two cases for each term proves the claim.
\end{proof}

The next lemmas provide tools to address the active vertices that are
not fully related to the path.  The lemmas themselves are purely
algebraic and do not directly consider the random choices.  In the
lemma, $\phi_1$ is meant to capture the potential of the subproblem,
whereas $\mphi_1$ captures the potential just with respect to those
vertices fully related to the path.  The motivation for these lemmas
is that the potential always decreases down to $\mphi$ regardless of
what pivot is selected.  If a fully related pivot gets selected, the
potential reduces even further.  The goal is to quantify the balance
between the cases, showing that the situation is at least as good as
the situation in which all vertices are fully path related.   When
\lemref{goodfunc} is applied, $q$ shall be used to mean the result of
\lemref{tightnode}, i.e., $q = (1/\sqrt{2} + O(1/C_\phi))$.  

\begin{lemma}\lemlabel{goodfunc}
  Let $g(\alpha,\beta,\delta)$ be any function over the
  number of ancestors, bridges, and descendents.  Let $\malpha$,
  $\mbeta$, and $\mdelta$ be counts that could arise by removing
  relationships between vertices, i.e., reducing any counts or moving
  bridges to ancestors/descendents.  Let $p =
  (\malpha+\mbeta+\mdelta)/(\alpha+\beta+\delta \leq 1$. 

  Suppose $g$ satisfies
  $g(\malpha,\mbeta,\mdelta) \leq \sqrt{p}\cdot
  g(\alpha,\beta,\delta)$
  for all valid input values.  Then for any $q \geq 2/3$,
  \[ p\cdot q\cdot g(\malpha,\mbeta,\mdelta) + (1-p)\cdot g(\malpha,\mbeta,\mdelta)
  \leq q\cdot g(\alpha,\beta,\delta) \ .\]
\end{lemma}
\begin{proof}
  Substituting in
  $g(\malpha,\mbeta,\mdelta) \leq \sqrt{p}\cdot
  g(\alpha,\beta,\delta)$
  gives
  $p\cdot q\cdot g(\malpha,\mbeta,\mdelta) + (1-p)\cdot
  g(\malpha,\mbeta,\mdelta)\leq pq\sqrt{p} \cdot
  g(\alpha,\beta,\delta) + (1-p)\sqrt{p}\cdot g(\alpha,\beta,\delta) =
  \sqrt{p}(pq - p + 1)g(\alpha,\beta,\delta)$.
  The claim follows as long as $\sqrt{p}(pq-p+1)\leq q$.  Treat $q$ as
  a constant and observe how the function of $p$ changes.  For
  $q=2/3$, the expression on the left is maximized at $p=1$, solving
  to exactly $q$.  As $q$ increases, the maximum of the function
  shifts even further to the right, meaning that the expression is
  still maximized for $p\in [0,1]$ at $p=1$.
\end{proof}

\begin{lemma}\lemlabel{goodfuncworks}
  The function $\phi$ satisfies the conditions of \lemref{goodfunc}. 
  Notably use $\phi$ to mean the potential applied to $\alpha$,
  $\beta$, and $\delta$, and $\mphi$ to mean the potential for
  $\malpha$, $\mbeta$, and $\mdelta$, where the remainder of the
  notation is the same as \lemref{goodfunc}.  Then $\mphi \leq
  \sqrt{p}\cdot \phi$.  
\end{lemma}
\begin{proof}
  Also define $\eta = \alpha+\beta+\delta$ and
  $\meta=\malpha+\mbeta+\mdelta$.  Use $\phi_1$ and $\phi_2$,
  respectively, to mean the potential applied to values $\alpha$,
  $\beta$, and $\delta$.  Similarly use $\mphi_1$ and $\mphi_2$ for
  $\malpha$, $\mbeta$, and $\mdelta$.

  Bound $\phi_1$ and $\phi_2$ separately.  The latter is trivial ---
  $\phi_2 = \eta$, so $\meta = p \eta$ implies $\mphi_2 = p \mphi_2
  \leq \sqrt{p}\mphi_2$ for $p \leq 1$.    The $\phi_1$ bound is a
  little
  harder because $\beta$ is double-counted in $\phi_1$.  

  To bound $\phi_1$, consider the following.  The fraction $p$
  dictates how much $\alpha$, $\beta$, and/or $\delta$ must be reduced
  by in total.  The worst case for the desired inequality is to
  maximize $\mphi_1$.  Without loss of generality, suppose
  $\alpha \leq \delta$, The $\mphi_1$ term is maximized when $\malpha$
  and $\mdelta$ are kept as large and as balanced as possible, so 
  $\delta$ should be reduced first.  If $p(\delta+\beta) > (\alpha+\beta)$,
  then only $\delta$ should be reduced, giving
  $\mdelta+\mbeta \leq p(\delta + \beta)$.  The potential thus becomes
  $\mphi_1 = \sqrt{(\malpha+\mbeta)(\mdelta+\mbeta)} \leq
  \sqrt{(\alpha+\beta)(p(\delta+\beta))} = \sqrt{p}\cdot \phi_1$.
  If $p$ is smaller, then consider two phases $p_1$ and $p_2$ with
  $p=p_1p_2$ and $p_1 = (\alpha+\beta)/(\delta+\beta)$.  During the
  first, the potential is maximized as above, leaving
  $\mphi_1 = \sqrt{p}\cdot \phi_1$ as above.  For the second phase,
  $\alpha=\delta$, so the best choice to keep the expression maximized
  is to balance the reductions from both counts simultaneously.  For
  this regime,
  $\mphi_1 = \sqrt{(\malpha+\mbeta)(\mdelta+\mbeta)} \leq
  \sqrt{p_2(\alpha+\beta)p_2(\delta+\beta)} = p_2 \phi_1$.
  Multiplying the two together gives a maximum value of
  $p_2\sqrt{p_1}\cdot \phi_1 \leq \sqrt{p_1p_2}\cdot \phi_1$.
\end{proof}

The next lemma pulls together all the pieces to argue that the
expected reduction on $\phi(s)$ is still almost as good as previously.

\begin{lemma}\lemlabel{partightbound}
  Consider any path-relevant subproblem $s=(G,P)$.  Let
  $s_1,s_2,\ldots$ be random variables denoting its child subproblems
  in the flattened path-relevant tree, and suppose
  $N_L \geq C_\phi^{2.5}\log n = \Omega(\log^6n)$.  Then
  $E[\phi(s_1) + \phi(s_2) + \cdots ] \leq (\phi(s)/\sqrt{2})(1 +
  O(1/\sqrt{C_\phi}))$.
\end{lemma}
\begin{proof}
  Let~$i$ denote the iteration during which at least one pivot $x_j$
  that is partly related to the path $P$ is selected.  Before that,
  the potential can only decrease.  Let
  $\alpha = \card{\AncD{G}{P}{\dmax D}}$ denote the number of partly
  path-related vertices in the graph $G$ at the start of the
  iteration.  Define $\beta$ and $\delta$ similarly for bridges and
  descendents, respectively.  Let $\malpha = \card{\AncD{G}{P}{\dmin
      D}}$ denote the number of fully path-related ancestors at the
  start.  Similarly for $\mbeta$ and $\mdelta$. 

  If any of the pivots is a bridge, there are no path-relevant
  subproblems (\lemref{parrewrite}) and the potential is 0.

  Suppose instead that no pivot beyond the first is a bridge, which
  can only increase the potential.  Then consider the pivots and
  corresponding partition steps in order, as in \lemref{parrewrite}.
  Let $s_1$ denote the recursive subproblem generated by the first
  pivot, and let $s_{1,R}$ denote the nonrecursive subproblem
  corresponding to vertices not found by a the core searches.  The
  second pivot partitions $s_{1,R}$ into recursive problem $s_2$ and
  remainder $s_{2,R}$.  The third pivot partitions $s_{2,R}$, and so
  on.  For $r$ pivots, the subproblems are
  $s_1,s_2,\ldots,s_r,s_{r+1}$, where $s_{r+1} = s_{r,R}$.

  The goal is to bound $E[\sum_{i=1}^{r+1} \phi(s_i)]$ given that
  there is at least one partly path-related pivot.  For each of the
  subproblems, let $\malpha_i$, $\mbeta_i$, and $\mdelta_i$ denote the
  number of fully path-related vertices that are preserved, i.e., part
  of core searches.  Let $\meta_i = \malpha_1+\mbeta_i+\mdelta_i$.
  Only the fully related vertices need be considered as these are the
  only ones that can be active at the next level.  Also consider the
  result $\malpha_{1,R}$, $\mbeta_{1,R}$ and $\mdelta_{1,R}$ of the
  first search. Let $f_i$ denote the number of path-active fringe
  nodes added, and let $f = \sum_{i=1}^{r+1}f_i$.  For simplicity,
  double-count the fringe vertices, giving us the following:
  \begin{align}
    \sum_{i=1}^{r+1} \phi(s_i) &\leq \phi(s_1) +
                                 \sum_{i=2}^{r+1}(\phi_1(s_i)+\phi_2(s_i))
    \nonumber\\ 
    &\leq \phi(s_i) +
                                 \left(\sqrt{(\malpha_i+\mbeta_i+f_i)(\mdelta_i+\mbeta_i+f_i)} +
                                 (\meta_i+f_i)\right)
    \nonumber\\
    &\leq
      \sqrt{(\malpha_1+\mbeta_1+f_1)(\mdelta_1+\mbeta_1+f_1)}+(\meta_1+f_1) \nonumber\\
&\hspace{.5in} + \sqrt{\left(\sum_{i=2}^{r+1}(\malpha_i
      +\mbeta_i + f_i)\right)\left(\sum_{i=2}^{r+1}(\mdelta_i +
      \mbeta_i + f_i)\right)} 
      +\sum_{i=2}^{r+1} (\meta_i+f_i) \tag*{(\lemref{phisplit})}\nonumber\\
    &\leq
      \left(1+\frac{1}{\sqrt{C_\phi}}\right) \Bigg(\sqrt{(\malpha_1+\mbeta_1)(\mdelta_1+\mbeta_1)}+(\meta_1) \nonumber\\
&\hspace{.5in} + \sqrt{\left(\sum_{i=2}^{r+1}(\malpha_i
      +\mbeta_i)\right)\left(\sum_{i=2}^{r+1}(\mdelta_i +
      \mbeta_i)\right)} 
      +\sum_{i=2}^{r+1} \meta_i\Bigg) + 3fC_\phi^2 \tag*{(\lemref{phifringe})}\nonumber\\
    &=
      \left(1+\frac{1}{\sqrt{C_\phi}}\right) \bigg(\sqrt{(\malpha_1+\mbeta_1)(\mdelta_1+\mbeta_1)}+\meta_1 \nonumber\\
&\hspace{1in} + \sqrt{(\malpha_{1,R}
      +\mbeta_{1,R})(\mdelta_{1,R} +
      \mbeta_{1,R})} 
      +\meta_{1,R}\bigg) + 3fC_\phi^2 \label{eqn}
\end{align}
At this point, the multiple pivots and the fringe nodes have all been
extracted from the main expression.  Let
$\mphi' = \sqrt{(\malpha_1+\mbeta_1)(\mdelta_1+\mbeta_1)} + \meta_1$ and
$\mphi'' =
\sqrt{(\malpha_{1,R}+\mbeta_{1,R})(\mdelta_{1,R}+\mbeta_{1,R})} +
\meta_{1,R}$.
The sum $\mphi' + \mphi''$ looks exactly like the random variables
generated from a single partition without fringe nodes.  Note that to
this point, no expectation has yet been applied; all manipulations
thus far are just algebra on the random variables.  \lemref{tightnode}
can thus be applied as long as the expectation is performed in a way
consistent with \lemref{parasymmetry}.

\lemref{parasymmetry} applies if the pivot is fully path related.  It
does not if the pivot is only partly path related.   There are thus 
cases depending on how the first pivot $x$ is classified.  Let $p =
(\malpha+\mbeta+\mdelta)/(\alpha+\beta+\delta)$ be the fraction of
partly path-related vertices that are fully related.  Let $F_x$ be the
event that $x$ is fully related with the path.  Let $\mphi =
\sqrt{(\malpha+\mbeta)(\mdelta+\mbeta)}+\malpha+\mbeta+\mdelta$.  
Then
\begin{align*}
  E\left[\mphi' + \mphi''\right] &= \prob{F_x} \cdot
                        E\left[\mphi' +\mphi'' | F_x\right] + (1-\prob{F_x})
                         \cdot E\left[\mphi'+\mphi''| \neg F_x\right]\\
  &\leq p \cdot E\left[\mphi' +\mphi'' | F_x\right] + (1-p) \cdot
    \mphi \tag*{(\lemref{phisplit})} \\
  &\leq p\cdot \left(\frac{1}{\sqrt{2}} +
    \frac{2}{\sqrt{C_\phi}}\right) \mphi + (1-p) \cdot
    \mphi \tag*{(\lemref{tightnode})}\\
  &\leq \left(\frac{1}{\sqrt{2}} +
    \frac{2}{\sqrt{C_\phi}}\right) \phi(s) \tag*{(\lemreftwo{goodfunc}{goodfuncworks})}
\end{align*}
Substituting back into Equation~\ref{eqn}  gives
\begin{align*}
 E\left[\sum_{i=1}^{r+1} \phi(s_i)\right] &\leq
\frac{1}{\sqrt{2}}\left(1+1/\sqrt{C_\phi}\right)\left(1+2\sqrt{2}/\sqrt{C_\phi}\right)
\phi(s) + E[3fC_\phi^2] \\
  &\leq \frac{1}{\sqrt{2}} \left(1+6/\sqrt{C_\phi}\right) \phi(s) +
    3C_\phi^2 \cdot E[f] \tag*{for $C_\phi\geq 4$} \\
  &\leq \frac{1}{\sqrt{2}} \left(1+6/\sqrt{C_\phi}\right) \phi(s) +
    3C_\phi^2 \cdot O((\malpha+\mbeta+\mdelta)\log n / N_L)
    \tag*{(\lemref{fringe})} \\
  &\leq \frac{1}{\sqrt{2}} \left(1+6/\sqrt{C_\phi}\right) \phi(s) +
    3C_\phi^2 \cdot O(\phi(s)\log n / N_L)\\
  &\leq \frac{1}{\sqrt{2}} \left(1+6/\sqrt{C_\phi}\right) \phi(s) +
    \phi(s) \cdot O(1/\sqrt{C_\phi}) \tag*{for $N_L \geq
    C_\phi^{2.5}\log n$} \\
  &= \frac{1}{\sqrt{2}}\left(1+O\left(\frac{1}{\sqrt{C_\phi}}\right)\right) \phi(s)
\end{align*}
\end{proof}

\subsection{Analyzing the Layers in the Tree}

Define the total potential $\Phi$ of a level as follows:
\[ \Phi(I_r) = (1+c_\Phi/\lg n)^{\lg n-r}\sum_{s \in I_r} \phi(s) \
,\]
where $I_r$ is the collection of subproblems corresponding to level
$r$ in the flattened tree and $c_\Phi$ is a constant to be set later.

\begin{corollary}
  Suppose $C_\phi = \Theta(\lg^2 n)$ and $N_L = \Omega(\lg^6 n)$ Then
  there exists a large-enough constant $c_\Phi$ such that
  $E[\Phi(I_r) | I_{r-1}] \leq \Phi(I_{r-1}) / \sqrt{2}$.
\end{corollary}
\begin{proof}
  Choose $c_\Phi$ large enough so that $1+c_\Phi/\lg n$ is greater than
  the $(1+O(1/\sqrt{C_\phi})) = 1+O(1/\log n))$ term in
  \lemref{partightbound}.  The claim then follows from linearity of
  expectation over subproblems.
\end{proof} 

The real purpose of the extra $(1+c_\Phi/\lg n)^{\lg n -r}$ factor is
to offset any potential increases to the subproblem potentials~$\phi$.
With an unlucky number of active fringe vertices, it is possible that
$\Phi$ increase when going from one row to the next.  Such an
increase, called a \defn{fringe failure}, would preclude the
application of \thmref{karp},  The next lemma shows that fringe
failures are unlikely.

\begin{lemma}\lemlabel{fringefail}
  There exist constants $c_\Phi$ and $c_L$ such that: for
  $C_\phi = \Theta(\lg^2 n)$ and
  $N_L \geq c_L C_\phi^3\lg n = \Omega(\log^7n)$,
  $\prob{\Phi(I_{r+1}) > \Phi(I_r)} \leq 1/(8\lg n)$.
\end{lemma}
\begin{proof}
  Even if all active vertices are preserved, \lemref{phisplit} states
  that the subproblem potentials $\sum_s \phi(s)$ cannot increase
  without the addition of fringe nodes.  The active fringe nodes
  themselves have two contributions (see \lemref{phifringe}): a
  multiplicative $(1+O(1/\sqrt{C_\phi}))$ overhead, and an additive
  $3C_\phi^2f$.  The former does not depend on the number of fringe
  nodes, so choose $c_\Phi/\lg n$ to be say twice as large as the
  $O(1/\sqrt{C_\phi})$ term.  Thus for $\Phi$ to increase would
  require that the total contribution from $f$ fringe nodes exceed
  $f \geq c_\Phi/(2\lg n) \sum_{s\in I_r}\phi(s)$.  For large enough
  $N_L$, the expected number of fringe nodes is
  $E[f] = O(\sum_{s \in I_r} \phi(s) \log n / N_L) \leq \sum_{s \in
    I_r} \phi(s) / (3C_\phi^3)$
  giving an expected potential contribution of
  $\sum_{s\in I_r} \phi(s) / C_\phi$.  For large enough $c_\Phi$ and
  $C_\phi=\Theta(\lg^2n)$, this expectation is at most
  $c_\Phi /(16\lg^2n) \sum_{s\in I_r} \phi(s)$.  Reaching the target
  threshold would require being $8\lg n$ times the expectation, which
  occurs with probability at most $1/(8\lg n)$ by Markov's inequality.
\end{proof}

To prevent $\Phi$ from increasing at all, instead define $\Phi'$ to be
equal to $\Phi$, except that it drops to 0 when a fringe failure
occurs.  It follows that (1) $\Phi'$ never increases, and (2)
$E[\Phi'(I_{r+1} | I_r] \leq E[\Phi(I_{r+1} | I_r]$.  \thmref{karp}
can now be applied.

\begin{lemma}\lemlabel{keypar}
  Let $\hat{G} = (\hat{V},\hat{E})$ be a directed graph, and let
  $n=\card{\hat{V}}$.  There exists a setting of
  $D=\Theta(n^{2/3}\log^{4/3}n)$, $C_\phi = \Theta(\log^2n)$,
  $N_L = \Theta(\log^7 n)$, and $\epsilon_\pi = O(1/\log^3n)$ such
  that the following holds.

  Consider any directed path $\hat{P}$ with $\length{\hat{P}} \leq D$ from
  $u$ to $v$.  Let $S$ be the shortcuts produced by an execution of
  \algref{par} on $\hat{G}$ with starting $h=\lg n$.  Then with
  probability at least 5/8: there exists a path from $u$ to $v$ in
  $G_S=(\hat{V}, \hat{E}\cup S)$ consisting of at most $D/2$ arcs.
\end{lemma}
\begin{proof}
  The starting value of
  $\Phi'(I_0) \leq (1+c_\Phi/\lg n)^{\lg n}(C_\phi + 1)n = O(n\log^2
  n)$.
  For large enough constant $w$, \thmref{karp} states that
  $\prob{\Phi'(I_{r+w}) > (1/\sqrt{2})^r O(n\lg n)} < 1/8$.  Then for
  $r=(2/3)\lg n + (4/3)\lg\lg n+\Theta(1)$, this expression reduces to
  $\prob{\Phi'(I_{r+w}) > c n^{2/3}\lg^{4/3}n} < 1/8$, for some
  constant $c$.

  If a fringe failure occurs, the bound on $\Phi'$ is meaningless.
  The probability of a fringe failure is at most the union bound over
  $r< \lg n$ levels of the failure probability $1/(8\lg n)$ from
  \lemref{fringefail}, which reduces to $1/8$.  If neither of these
  failures occurs, the bound on $\Phi'$ implies a bound on active
  unshortcutted subpaths, as all path vertices counted as bridges
  towards $\Phi'$.  Thus, with failure probability $1/4$, the total
  length of all subpaths in path-relevant subproblems in level-$(r+2)$
  is at most $O(n^{2/3}\lg^{4/3}n)$.
  
  Finally, consider the concatenations and shortcutted leaves via
  \lemref{partree}.  With failure probability $1/8$, the total
  shortcutted length is thus $O(n^{2/3}\lg^{4/3}n)$.  Choose $D$ to be
  a constant factor larger than the constant hidden inside the big-O.
\end{proof}

\subsection{Completing the Proof of \lemref{parmain}}

This section proves bounds on the number of shortcuts and overall work
performed.  The main goal is to show that, with probability at least
$7/8$, the number of vertices (and hence shortcuts) and arcs visited
by searches is consistent with \lemref{keypar}.  Thus, with
probability at least $1/2$, both \lemref{keypar} and this bound hold.
Combining with \lemref{keypar}, this yields a proof of
\lemref{parmain} and hence also \thmref{parfull}.

The settings used are $D=\Theta(n^{2/3}\log^{4/3}n)$,
$C_\phi = \Theta(\log^2n)$, $N_L = \Theta(\log^7n)$,
$\epsilon_\pi = \Theta(1/\log^3n)$, and $N_k = \Theta(\log^4n)$, as
dictated by constraints offered in previous lemmas.   The maximum
search distance is immediate: it is at most $hDN_kN_L = O(\log n \cdot
n^{2/3}\log^{4/3}n\cdot \log^4n \cdot \log^7n) = O(n^{2/3}\log^{14}n)$.

Consider each level of recursion in \algref{par}.  \corref{numvisits}
holds with high probability, so assume that it holds for every node at
every iteration.  Consider any vertex in the iteration in which it is
visited by a core search.  By assumption of \corref{numvisits}, the
vertex, and hence its incident arcs, is visited by at most $O(\log n)$
searches. 

The number of vertices (and hence arcs) may increases with each level
in the tree due to fringe searches.  The final step of the proof is to
argue that the total size of all subproblems in the final level of
recursion is $O(n)$ vertices and $O(m)$ arcs, and hence the total cost
of all levels is $O(n\log^2n)$ vertices and $O(m\log^2n)$ arcs.  

By \lemref{fringe} with $N_L=\log^7n$, the number of vertices and arcs
increases with each level of recursion by an additive $O(n'/\log^2n)$
and $O(m'/\log^2n)$ in expectation, where $n'$ and $m'$ are the
current numbers at that level.  Thus, by Markov's inequality, with
probability at most $1/(8\lg n)$, the increases is not more than a
multiplicative $1+O(1/\log n)$.  Since there are $\lg n$ levels of
recursion, this results in probability at most $1/8$ of exceeding a
total of $n(1+O(1/\log n))^{\lg n} = O(n)$ vertices and
$m(1+O(1/\log n))^{\lg n} = O(m)$ arcs.

\section{Parallel Version}\seclabel{details}

This section analyzes a parallel version of \algref{par} and
\algref{pardiam}.  This section assumes the reader is comfortable
enough with parallel algorithms to infer the details, instead focusing
only on the interesting issues.

The main results are as follows. 
\begin{theorem}\thmlabel{pardiamalgbounds}
  There exists a randomized parallel algorithm taking as input a
  directed graph $\hat{G}=(\hat{V},\hat{E})$ with the following
  guarantees.  Let $n=\card{\hat{V}}$, $m=\card{\hat{E}}$, and without
  loss of generality assume $m\geq n/2$. Then (1) the algorithm
  produces a size-$O(n\log^4n)$ set $S^*$ of shortcuts; (2) the
  algorithm has $O(m\log^6n+n\log^{10}n)$ work; (3) the algorithm has
  $O(n^{2/3}\log^{21}n)$ span; and (4) with high probability, the
  diameter of $G_{S^*} = (\hat{V}, \hat{E}\cup S^*)$ is
  $O(n^{2/3}\log^{4/3}n)$.  
\end{theorem}

\begin{corollary}
  There exists a randomized parallel CREW algorithm for digraph
  reachability that has work $O(m \log^6n+n\log^{10}n)$ work and
  $O(n^{2/3}\log^{21}n)$ span, both with high probability.
\end{corollary}
\begin{proof}
  Perform the diameter reduction algorithm, then run a standard
  parallel BFS but limited to $O(n^{2/3}\log^{4/3})$ hops.  The work
  and span of the diameter reduction dominates.  If the BFS completes
  in the prescribed number of rounds, the algorithm terminates.
  Otherwise, keep repeating the diameter reduction and BFS until
  successful. 
\end{proof}

\paragraph{Model.} This paper adopts the \emph{de facto} standard
\defn{work-span model}~\cite{CLRS}, also called
work-time~\cite{JaJa92} or work-depth model, which abstracts low-level
details of the machine such as the number of processors or how
parallel tasks are scheduled.  The work-span model allows algorithms
to be expressed through the inclusion of parallel loops, i.e., a
parallel foreach. A parallel foreach indicates that each task
corresponding to a loop iteration may execute in parallel, and that
all parallel tasks must complete before continuing to the next step
after the loop.  It is generally straightforward to map algorithms
from the work-span model to a PRAM model; see,
e.g.,~\cite{JaJa92,KarpRa91}.  Like the asynchronous PRAM
model~\cite{Gibbons89}, the work-span model requires that algorithmic
correctness not be tied to any assumptions about how tasks are
scheduled beyond the explicit ordering imposed by the loops.  That is
to say, it should not be assumed that the instructions across
iterations execute in lock step.

The \defn{work} of an algorithm is the same as the sequential running
time in a RAM model if all parallel loops are replaced by sequential
loops.  When multiple tasks are combined through a parallel loop, the
\defn{span} of the composition is the maximum of the span of the
individual subproblems, plus the span of the loop itself.  There are
several variants to the work-span model. In a \defn{binary-forking
  model} such as~\cite{CLRS}, the span of a $k$-way loop is
$\Theta(\lg k)$.  Much of the literature on parallel algorithms,
however, adopts an \defn{unlimited-forking model}, where the span of
launching $k$ parallel tasks adds $O(1)$ to the span.  Since many of
the subroutines employed are analyzed in this model, this paper adopts
the unlimited-forking model.  PRAM algorithms, for example, correspond
to an unlimited forking model.  Both models only differ by logarithmic
factors in the span.

The algorithm is a concurrent-read exclusive-write (\defn{CREW}) algorithm.
CREW means that multiple parallel tasks may read the same data, but
they may not write to the same location.\footnote{CREW is usually a
  restriction applied to the
  PRAM~\cite{FortuneWy78,Goldschlager78,SavitchSt79} machine model,
  e.g., a CREW PRAM.  In contrast, the work-span model is an
  algorithmic cost model, not a machine model.  This paper proposes
  lifting the CREW qualifier to the work-span level rather than the
  PRAM level.}

\subsubsection*{Performing Concurrent Searches}

The key subroutine in \algref{par} are the $dD$-limited searches to
find, e.g., $R^+_j$.  One might simply replace the foreach loops by
parallel loops, but the question is how the bookkeeping should be
performed.  Ordinarily, a BFS keeps track of already-visited vertices
by either annotating vertices in the graph directly, or equivalently
by keeping an extra array indexed by vertex. A natural way to perform
multiple searches in parallel using a CREW algorithm would thus be to
duplicate the bookkeeping efforts for each parallel search, but doing so would
increase the work dramatically just to copy the graph or initialize
the arrays.

The key property that allows an efficient implementation is
\corref{numvisits} --- with high probability, no vertex is visited by
more than $O(\log n)$ parallel searches.  The implementation may
assume that this is the case, and just abort by returning immediately
if a vertex gets visited too many times.

The main goal is to support the following for each call to \foopar.

\begin{lemma}\lemlabel{parsearches}
  Consider an iteration $i$ in call to \foopar on graph $G=(V,E)$.
  Let $n_e$ be the total number of arcs traversed by searches, counting
  each  arc once per search that reaches it.  Then an iteration can be
  implemented with $O(n_e \log^2n + \card{X_i}\log n)$ work and
  $O(n^{2/3}\log^{15})$ span.
\end{lemma}

The remainder of the section is devoted to exhibiting an algorithm
that proves \lemref{parsearches}.

The set of searches from $X_i$ (in one direction) are grouped together
as a single modified BFS.  Rather than marking a vertex with a single
bit indicating whether it has been discovered, a vertex is tagged with
a list of IDs of the pivots that have reached it. Every time this list
of IDs changes, the vertex may be re-added to the frontier and all of
its outgoing arcs explored again.  Since a vertex is not visited too
many times, the overhead is not too high.
 
In more detail, the algorithm is as follows.  At the start of the call
to \foopar, initialize $\Theta(\log n)$ space for each vertex to
record the ID tags, initally all null.  Use an array to store the
frontier vertices along with the ID of the pivot from which this
search originated; a vertex may appear in the frontier multiple times
from different pivots.  Save all frontiers so as to identify all
vertices reached by the searches at the end and also to record all new
shortcuts.

To start a set of searches from $\card{X_i}$, copy all live pivots
$x_j$ to the frontier array and associate with each pivot its own ID
as the originator of the search.  Also update each pivot's tag list to
include itself.

Each round of the BFS operates as follows.  Foreach vertex in the
frontier in parallel, identify the number of outgoing arcs.  Next,
perform parallel prefix sums so that each arc has a distinct index in
the next frontier array.  Foreach arc $(u,v)$ in parallel, let $x_j$
be the associated pivot ID.  Check whether $v$'s ID set includes
$x_j$; this check can be performed in $O(\log n)$ sequential time
(both work and span) by scanning through $v$'s tag list.  If $x_j$ is
not present, record $v$ and $x_j$ in $(u,v)$'s slot in the next
frontier; otherwise record null.  

At this point, a vertex may appear many times in the frontier list,
even from a single search.  Sort the frontier list by vertex (high
priority) and pivot ID (lower priority).  Remove duplicate entries
with a compaction pass.  Now each vertex appears at most once for each
search, so $O(\log n)$ times in total.  For each slot $j$ in the next frontier
in parallel, let $v$ be the vertex stored there.  Check whether this
is the first slot for vertex $v$, i.e., if $j-1$ stores a different
vertex.  If so, scan through the $O(\log n)$ next slots
(sequentially), and for each entry of $v$ append the pivot tag to
$v$'s tag list.

Repeat this process for the number of rounds dictated by the distance
$dD$ for the core searches.  Extend the search another $D$ rounds for
the fringe search, but use a different tag list and frontier array
going forward. 

When the searches complete, sort the arrays of all vertices reached by
core searches and by fringe searches.  Foreach vertex $v$ in core
searches, in parallel, identify the lowest ID pivot reaching~$v$.  Again
use parallel prefix sums and then copy the lowest-ID occurrence of $v$
to a new array for the recursive searches.  Finally, sort the new
array by pivot ID so that all vertices in the same induced subgraph
are adjacent.  Building the induced subgraphs for recursive calls can
again be accomplished with arc counting, prefix sums, and sorting.

\paragraph{Updating $G[V_R]$.} One could build $G[V_R]$ explicitly,
but doing so would require processing the full graph.  The goal
expressed by \lemref{parsearches} is to have work proportional to the
number of arcs reached, but $G[V_R]$ could be much larger.  Instead,
simply mark vertices in $V$ as dead when they have been reached by a
core search.  Augment the search to ignore dead vertices.

\vspace{1em}
\noindent\emph{Completing the proof of \lemref{parsearches}.}
The basic subroutines used in each round such as prefix sums,
compaction, etc, can all be performed in linear work and $O(\log n)$
span. (See e.g.,~\cite{JaJa92}.)  Scanning the list of tags also
requires $O(\log n)$ work per arc on the frontier and $O(\log n)$ span
as it is performed sequentially.  Using Cole's merge
sort~\cite{Cole88}, the cost of a sort is $O(\log n)$ work per element
sorted and $O(\log n)$ span.  Multiplying the search distance by
$O(\log n)$ thus gives the overall span bound.  Since each arc may be
reached by $O(\log n)$ searches, the bound is $O(\log^2n)$ work per
arc visited.\qed

\subsubsection*{Aborting \algref{par}}

To make the work (and shortcut) bound deterministic, \algref{pardiam}
needs the ability to abort any runs of \algref{par} that exceed the
target work bound. (Exceeding the shortcut bound can be handled simply
discarding the result --- a true abort is not necessary there.)

Unfortunately, the proof of \lemref{parmain} examines the work in
aggregate across levels in the recursion tree.  It is likely that
individual recursive subproblems will do more work, so abort decisions
are not local.

One simple option is to augment the algorithm to check the elapsed
time, and to return immediately if some threshold has been reached.
Technically, however, this solution violates the work-span model as
the target time bound would depend on both on how efficiently the
program is scheduled and on the number of processors employed.  

There is a solution in the work-span model --- logically implement the
recursive steps of the algorithm as a BFS.  That is, maintain an array
of subproblems, initially just \foopar{$\hat{G},\lg n$}.  To implement
a level of recursive, perform prefix sums to add up the total number
of vertices across all subproblem, and give each vertex (pivot) a
specific slot to put its recursive subproblem.  Instead of launching
the recursive subproblems immediately, simply record them in the
appropriate slot.  When all subproblems at this level of recursion
complete, launch all problems at the next level (in parallel).   

The work bound can only be exceeded if the total number of arcs in the
next set of subproblems grows too large.  This number can be counted
with a parallel reduce after each level of recursion completes.  None
of these steps increase the work or span asymptotically.

\paragraph{Proof of \thmref{pardiamalgbounds}.}  The diameter and
shortcut bounds are directly from \thmref{parfull}.  Multiplying the
cost per arc of \lemref{parsearches} with the number of arcs searched
in \thmref{parfull} gives the work bound.  Shortcuts can be gathered
and larger graphs built for each iteration of \algref{pardiam} by
sorting, but that is dominated by the other work performed.

The span bound is obtained by multiplying the maximum search distance
of $O(n^{2/3}\log^{14}n)$ by the $O(\log n)$ span per BFS round, the
$N_k = O(\log^4n)$ iterations in a call, the $O(\log n)$ levels of
recursion in a run of \algref{par}, and the $O(\log n)$ iterations of
the outer loop of \algref{pardiam}.  Note that the inner loop of
\algref{pardiam} can be implemented in parallel.  All together, that
gives $O(n^{2/3}\log^{21}n)$ span.  \qed

\section{Building a Directed Spanning Tree}\seclabel{tree}

This section discusses how to augmented the algorithm to produce a
directed spanning tree.  It is not immediately obvious how to do so
even for the sequential algorithm of \secref{sequential}.  To
illustrate the issues, consider the following graph:
$s\rightarrow u \rightarrow v \rightleftarrows w$.  If $w$ is selected
as a pivot first, then a shortcut $(s,w)$ is added and a BFS from $s$
in the shortcutted graph may discover the following path:
$s \leadsto w \rightarrow v$.  Simply splicing in the path
corresponding to the shortcut would result in $s\rightarrow u
\rightarrow v \rightarrow w \rightarrow v$, which is no longer a
simple path.  The goal is to do this splicing, but in a way that
avoids repeated vertices.  The situation is slightly more challenge in
the case of \algref{pardiam} because the arcs shortcutted could
themselves be shortcuts, but the result is just that several
iterations are needed.

The algorithm for building the directed spanning tree is a
postprocessing step performed after the full execution of
\algref{pardiam}.  The algorithm references the BFS trees used to
build shortcuts, however, so all BFS trees need to be saved as
\algref{pardiam} executes.  Each shortcut must also be augmented with
a reference to the BFS tree that produced it.   The forward-search BFS
trees are directed out from the root, whereas the backward-search BFS
trees are directed towards the root.  In this way, the BFS trees
correspond to arcs in some graph.

Let $G_0, G_1, \ldots, G_{k=\Theta(\log n)}$, where $G_0 = \hat{G}$,
denote the sequence of graphs built after each iteration of the outer
loop of \algref{pardiam}.  Running BFS on the resulting graph $G_k$
yields a directed spanning tree $T_k$ in $G_k$.  This section
describes how to transform a directed spanning tree $T_i$ in graph
$G_i$ to a directed spanning tree $T_{i-1}$ in $G_{i-1}$. Iterating
$\Theta(\log n)$ times gives a spanning tree in the original graph.

Start each iteration by labeling every vertex $v$ in the tree with
label $\id{low}(v)=0$ and $\id{high}(v)$, where $\id{high}(v)$ is
$v$'s distance from the root in $T_i$.  This can be accomplished in
linear work and logarithmic span using the Euler-tour
method~\cite{TarjanVi84}.

For the next step, the shortcuts in both directions are treated
differently.  The goal is to essentially splice in the paths, which
results in vertices appearing multiple times.  This multiplicity will
be resolved afterwards. 

For each vertex $u \in T_i$ in parallel, traverse all forward-search
BFS trees created in iteration~$i$ and rooted at $u$.  Label those
vertices $v$ with $\id{high}(v) = \id{high}(u)$, and $\id{low}(v)$ is
$v$'s depth (or distance from $u$) in the tree.  Note that these
labelings should be performed on the BFS trees themselves, not on
$T_i$ or $G_i$, as each vertex may belong to multiple trees and may
otherwise be labeled multiple times concurrently.\footnote{The
  Euler-tour technique could be applied to each tree, but a parallel
  BFS is sufficient here as the trees have depth
  $O(n^{2/3}\log^{4/3}n)$ by construction; the work and span would be
  at most the work and span of constructing the tree in the first
  place.}

For the backward direction, consider all arcs $(u,v)$ in $T_i$ in
parallel.  If $(u,v)$ is a shortcut on a backward-search BFS tree
rooted at $v$, traverse the path from $u$ to $v$ in the BFS tree and
label each vertex $w$ on the path by $\id{high}(w) = \id{high}(u)$.
Also label $\id{low}(w)$ with $w$'s distance from $u$ on the path. 

Finally, sort all arcs $(u,v)$ in the collection of BFS trees
traversed in the above process, as well as the arcs in $T_i$ that also
exist in $G_{i-1}$, by three values: $v$'s ID (most significant),
$\id{high}(u)$, and $\id{low}(u)$ (least significant). Foreach arc
$(u,v)$ in the sorted list in parallel, if it is the first arc
directed toward $v$ in sorted order then include the arc $(u,v)$ in
$T_{i-1}$.

\begin{lemma}
  Suppose that $T_i$ is a directed spanning tree in $G_i$ rooted at
  vertex $s$.  Then the $T_{i-1}$ produced is a directed spanning tree
  rooted at vertex $s$ consisting of only arcs present in $G_{i-1}$.
\end{lemma}
\begin{proof}
  Since only arcs present in $G_{i-1}$ are considered in the last step
  of the algorithm, $T_{i-1}$ is a subgraph of $G_{i-1}$.  It is not
  obvious, however, that it is a tree, nor is it obvious that it
  spans.  

  The first step is to show that every vertex, except $s$, has an
  incoming arc in $T_{i-1}$.  Consider a vertex $v$ and its incoming
  arc $(u,v)$ in $T_i$.  If $(u,v)$ is present in $G_{i-1}$ as well,
  then it is in consideration the last step, so $v$ must select an
  arc.  If $(u,v)$ is a shortcut, then it corresponds to some path in
  a BFS tree.  All arcs in that path, and specifically the arc
  directed toward $v$, are also in consideration.  Thus, $v$ has an
  incoming arc. 

  For each vertex, let the final label be the lowest label associated
  with any of its copies. If all arcs go from lower label to higher
  label, then there are no cycles.  To prove this is the case, the
  claim is that every copy of each vertex other than the source (and
  in particular the lowest-label copy) has an incoming arc from a
  vertex with a lower label.  Since the minimum incoming arc is
  the one used, that would imply that all arcs are from lower to
  higher label. 

  To prove the claim, consider a copy of vertex $v$.  There are three
  cases.  If $v$ is in a forward-search BFS tree and not the root,
  then $v$ has depth (and hence $\id{low}(v)$ label) one higher than
  its parent in the tree.  If $v$ is in a backward-search path and not
  the source, the same argument holds.

  Otherwise, $v$'s label is the same as in $T_i$.  In $T_i$, $v$'s incoming
  arc $(u,v)$ satisfies $\id{high}(u) < \id{high}(v)$ by
  construction.  If $(u,v)$ is in $G_{i-1}$, then this arc satisfies
  the claim.  Otherwise, $v$ the non-root of a BFS tree with a
  strictly lower $\id{high}$ value, and hence one of the first two
  cases applies.
\end{proof}
\vspace{-.5em}\section{Conclusions}
\vspace{-.5em}
This work makes the first major progress toward work-efficient parallel algorithms for
directed graphs, but it also exposes several new questions.  First, can
the performance be improved?  Shaving logarithmic factors would be
nice, but doing so seems premature --- it is quite likely that
$\Otilde{n^{2/3}}$ is not the final answer.  I would conjecture that
an $n^{1/2 + o(1)}$-diameter reduction is possible using a more
sophisticated algorithm based on the one presented herein.

Is true work efficiency, i.e., $O(m)$ work, possible for the
diameter-reduction problem?  Achieving that
would require first producing an $O(m)$-time sequential algorithm for
the problem.  

Hesse's lower bound provides a lower bound on work-efficient diameter
reduction, but that is not a general lower bound on digraph
reachability.  Can digraph reachability be improved by relaxing the
shortcutting requirements, perhaps by adopting some ideas from
Spencer's algorithm?   Are there good general lower bounds for
work/span tradeoffs of these algorithms? 

Finally, can the algorithm be extended to solve unweighted shortest paths?

\section*{Acknowledgements}
Special thanks to Cal Newport and Justin Thaler for some useful
discussions.  This work is supported in part by NSF grants
CCF-1718700, CCF-1617727, and CCF-1314633.

\bibliographystyle{plain}
\bibliography{reach}

\end{document}